\documentclass[fontsize=11pt,a4paper,DIV=17]{scrartcl}

\usepackage[utf8]{inputenc}
\usepackage{amsthm,amsmath,amsfonts}
\usepackage{thm-restate}
\usepackage{graphicx}
\usepackage{float}
\usepackage{enumerate}
\usepackage{hyperref}
\usepackage{xcolor}
\usepackage{bibitemsep}
\usepackage{complexity}
\usepackage{nicefrac}

\usepackage[font=small,labelfont=bf]{caption}
\usepackage[font=small,labelfont=normalfont,labelformat=simple]{subcaption}

\input{psfig.sty}
\graphicspath{{fig/}{../notes/fig/}}

\newtheorem{theorem}{Theorem}
\newtheorem{lemma}[theorem]{Lemma}
\newtheorem{fact}{Fact}
\newtheorem{corollary}[theorem]{Corollary}
\newtheorem{proposition}[theorem]{Proposition}

\newtheorem{problem}{Problem}
\newtheorem{conjecture}{Conjecture}

\newtheorem*{claim*}{Claim}

\newtheorem*{definition*}{Definition}

\def\RR{\mathbb{R}}
\def\SS{\mathbb{S}}

\def\AA{{\cal A}}
\def\LL{{\cal L}}
\def\CC{{\cal C}}
\def\EE{{\cal E}}
\def\AAfiveA{{\cal N}_5^1}  %%{\cal \hat{A}}_6}
\def\AAfiveB{{\cal N}_5^2}  %%{\cal \hat{A}}_6}
\def\AAfiveC{{\cal N}_5^3}  %%{\cal \hat{A}}_6}
\def\AAfiveD{{\cal N}_5^4}  
\def\AAsixA{{\cal N}_6^\Delta}  %%{\cal \hat{A}}_6}
\def\AAsixB{{\cal N}_6^2}  
\def\AAsixC{{\cal N}_6^3}  
\def\AAsixER{{\cal N}_6^{ER}}
\def\AAsixMiscIntB{{\cal N}_{6}^{i6:2}}
\def\AAsixMiscIntC{{\cal N}_{6}^{i6:3}}
\def\AAsixMiscConVierA{{\cal N}_{6}^{c4:1}}
\def\AAsixMiscConVierB{{\cal N}_{6}^{c4:2}}
\def\AAsixMiscConAchtA{{\cal N}_{6}^{c8:1}}
\def\AAsixMiscConAchtB{{\cal N}_{6}^{c8:2}}
\def\AAsixMiscConVierundZwanzig{{\cal N}_{6}^{c24}}

\def\ITEMMACRO #1 ??? #2 ???{\smallskip\par\noindent
\hangindent=#2em\setbox0\hbox{#1 \kern5pt}%
\ifdim\wd0<\hangindent\setbox0\hbox to\hangindent{\hss#1\quad}\fi%
\box0\ignorespaces}

\def\Item(#1){\ITEMMACRO {\rm (#1)} ??? 1.8 ???}

\def\BrackItem[#1]{\ITEMMACRO [#1] ??? 1.8 ???}

\title{Arrangements of Pseudocircles:\\ On~Circularizability\footnote{
Partially supported by the DFG Grants FE~340/11-1 and FE~340/12-1.
Manfred Scheucher was partially supported by the ERC Advanced Research Grant no.~267165 (DISCONV).
We gratefully acknowledge the computing time granted 
by TBK Automatisierung und Messtechnik GmbH
and by the Institute of Software Technology, Graz University of Technology.
We also thank Günter Rote, Boris Springborn, and the anonymous reviewers for valuable comments. 
}}
\newcommand*\samethanks[1][\value{footnote}]{\footnotemark[#1]}

\author{Stefan Felsner\thanks{Institut f\"ur Mathematik, Technische Universit\"at Berlin, Germany, 
\texttt{\{felsner,scheucher\}@math.tu-berlin.de}}
\and Manfred Scheucher\samethanks}

\begin{document}
\maketitle

\begin{abstract}
  An arrangement of pseudocircles is a collection of simple closed
  curves on the sphere or in the plane such that any two of the curves
  are either disjoint or intersect in exactly two crossing points.
  We call an arrangement intersecting if every pair of pseudocircles
  intersects twice. An arrangement is circularizable if there is a
  combinatorially equivalent arrangement of circles.

  In this paper we present the results of the first thorough study of
  circularizability. We show that there are exactly four non-circularizable
  arrangements of $5$ pseudocircles (one of them was known before).
  In the set of 2131 digon-free intersecting arrangements
  of $6$ pseudocircles we identify the three non-circularizable
  examples. We also show non-circularizability of 8 additional 
  arrangements of $6$ pseudocircles which have a group of symmetries
  of size at least~4.
  
  Most of our non-circularizability proofs depend on incidence theorems like
  Miquel's. In other cases we contradict circularizability by considering a
  continuous deformation  
  where the circles of an assumed circle representation grow
  or shrink in a controlled way.

  The claims that we have all non-circularizable arrangements
  with the given properties are based on a program that generated all
  arrangements up to a certain size. Given the complete lists of arrangements,
  we used heuristics to find circle representations. Examples where the
  heuristics failed were examined by hand.
\end{abstract}

\section{Introduction}
\label{sec:introduction}

Arrangements of pseudocircles generalize arrangements of circles in
the same vein as arrangements of pseudolines generalize arrangements
of lines. The study of arrangements of pseudolines was initiated by
Levi~\cite{Levi26} in~1926. Since then arrangements of pseudolines
were intensively studied. The handbook article on the
topic~\cite{FelsnerGoodman2016} lists more than 100 references. To the
best of our knowledge the study of arrangements of pseudocircles was
initiated by Gr\"unbaum~\cite{Gr72} in the 1970s.

A \emph{pseudocircle} is a simple closed curve in the plane or on the
sphere. An \emph{arrangement of pseudocircles} is a collection of
pseudocircles with the property that the intersection of any two of the
pseudocircles is either empty or consists of two points where the curves
cross. Other authors also allow touching pseudocircles,
e.g.~\cite{anppss-lenses-04}. 
The \emph{(primal) graph of an arrangement} $\AA$ of
pseudocircles has the intersection points of pseudocircles as
\emph{vertices}, the vertices split each of the pseudocircles into
arcs, these are the \emph{edges} of the graph. Note that this graph
may have multiple edges and loop edges without vertices. 
The graph of an arrangement of pseudocircles
comes with a plane embedding, the faces of this embedding are the
\emph{cells} of the arrangement. 
A cell of the arrangement with $k$ crossings on
its boundary is a \emph{$k$-cell}. A 2-cell is also called a \emph{digon}
(some authors call it a \emph{lens}), and a 3-cell is also called a
\emph{triangle}.  An arrangement $\AA$ of pseudocircles is
\begin{description}
\item[\emph{simple},] if no three pseudocircles of $\AA$
  intersect in a common point;
\item[\emph{connected},] if the graph of the arrangement is connected;
\item[\emph{intersecting},] if any two pseudocircles of $\AA$
  intersect. 
\item[\emph{cylindrical},] if there are two cells of the
  arrangement $\AA$ which are separated by each of the pseudocircles.
\end{description}
Note that every intersecting arrangement is connected. 
In this paper we assume that arrangements are simple and connected. 
The few cases where these assumptions do not hold will be clearly marked.

Two arrangements $\AA$ and $\mathcal{B}$ are \emph{isomorphic}
if they induce homeomorphic cell decompositions of the compactified plane,
i.e., on the sphere. 
Stereographic projections can be used to map
between arrangements of pseudocircles in the plane and arrangements of
pseudocircles on the sphere. Such projections are also considered
isomorphisms.  
In particular, the isomorphism class of an
arrangement of pseudocircles in the plane is closed under changes of
the unbounded cell. 

Figure~\ref{fig:all3-conn} shows the three connected 
arrangements of three pseudocircles 
and Figure~\ref{fig:all4} shows the 21 connected 
arrangements of four pseudocircles. 
We call the unique digon-free
intersecting arrangement of three (pseudo)circles 
the \emph{Krupp}\footnote{This name refers to the logo of the Krupp AG, a
  German steel company. Krupp was the largest company in Europe at the
  beginning of the 20th century. 
%   There is also a disease with the German name Pseudo-Krupp, we have no corresponding arrangement.
  }.
The second intersecting arrangement %of three pseudocircles 
is the \emph{NonKrupp}; this arrangement has digons.
The non-intersecting arrangement is the \emph{3-Chain}.

%%%%%%%%%%%%%%%%%%%%%%%%%%%%%%%%%%%%%%%%%%%%%%%%%%%
%%
% in einem figure environment mit caption
   \calc_figscale{18}
    \begin{figure}[htb]
    \centerline{\input{\path/all3-conn.pstex_t}}
    \caption{\label{fig:all3-conn}}
    \end{figure}
    VC
{The 3 connected arrangements of $n=3$ pseudocircles.
 (a)~\emph{Krupp}, (b)~\emph{NonKrupp}, (c)~\emph{3-Chain}.}
%%
%%%%%%%%%%%%%%%%%%%%%%%%%%%%%%%%%%%%%%%%%%%%%%%%%%%

%%%%%%%%%%%%%%%%%%%%%%%%%%%%%%%%%%%%%%%%%%%%%%%%%%%%%%%%%%%%%%%%%%%
%%%%%%%%%%%%%%%%%%%%%%%%%%%%%%%%%%%%%%%%%%%%%%%%%%%%%%%%%%%%%%%%
\begin{figure}[htp]
  \centering
  
  \begin{subfigure}[t]{.2\textwidth}
    \includegraphics[page=1,width=0.95\textwidth]{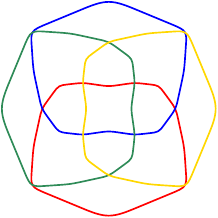}
    \caption{}
    \label{fig:all4_a}  
  \end{subfigure}
  \hfill
  \begin{subfigure}[t]{.2\textwidth}
    \includegraphics[page=1,width=0.95\textwidth]{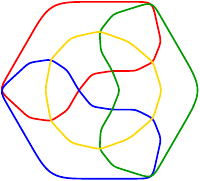}
    \caption{}
    \label{fig:all4_b}
  \end{subfigure}
  \hfill
  \begin{subfigure}[t]{.2\textwidth}
    \includegraphics[page=1,width=0.95\textwidth]{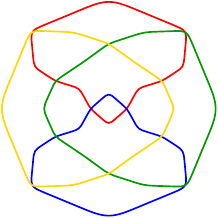}
    \caption{}
    \label{fig:all4_c}  
  \end{subfigure}
  \hfill
  \begin{subfigure}[t]{.2\textwidth}
    \includegraphics[page=1,width=0.95\textwidth]{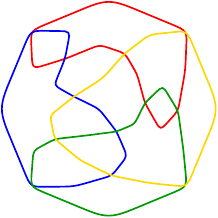}
    \caption{}
    \label{fig:all4_d}  
  \end{subfigure}

  \begin{subfigure}[t]{.2\textwidth}
    \includegraphics[page=1,width=0.95\textwidth]{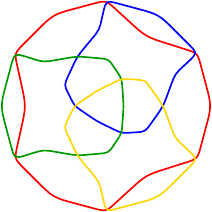}
    \caption{}
    \label{fig:all4_e}
  \end{subfigure}
  \hfill
  \begin{subfigure}[t]{.2\textwidth}
    \includegraphics[page=1,width=0.95\textwidth]{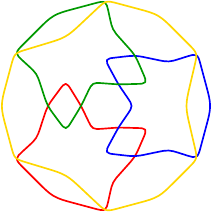}
    \caption{}
    \label{fig:all4_f}  
  \end{subfigure}
  \hfill
  \begin{subfigure}[t]{.2\textwidth}
    \includegraphics[page=1,width=0.95\textwidth]{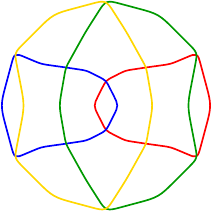}
    \caption{}
    \label{fig:all4_g}  
  \end{subfigure}
  \hfill
  \begin{subfigure}[t]{.2\textwidth}
    \includegraphics[page=1,width=0.95\textwidth]{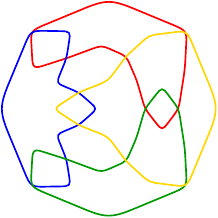}
    \caption{}
    \label{fig:all4_h}  
  \end{subfigure}

  \begin{subfigure}[t]{.2\textwidth}
    \includegraphics[page=1,width=0.95\textwidth]{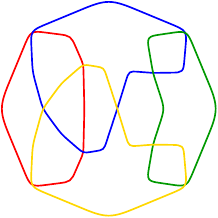}
    \caption{}
    \label{fig:all4_i}
  \end{subfigure}
  \hfill
  \begin{subfigure}[t]{.2\textwidth}
    \includegraphics[page=1,width=0.95\textwidth]{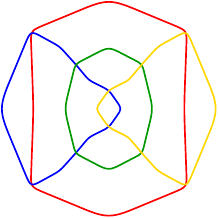}
    \caption{}
    \label{fig:all4_j}  
  \end{subfigure}
  \hfill
  \begin{subfigure}[t]{.2\textwidth}
    \includegraphics[page=1,width=0.95\textwidth]{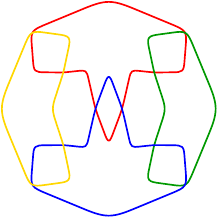}
    \caption{}
    \label{fig:all4_k}  
  \end{subfigure}
  \hfill
  \begin{subfigure}[t]{.2\textwidth}
    \includegraphics[page=1,width=0.95\textwidth]{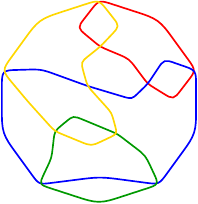}
    \caption{}
    \label{fig:all4_l}  
  \end{subfigure}

  \begin{subfigure}[t]{.2\textwidth}
    \includegraphics[page=1,width=0.95\textwidth]{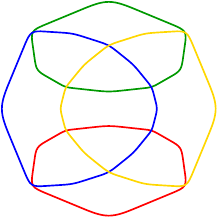}
    \caption{}
    \label{fig:all4_m}  
  \end{subfigure}
  \hfill
  \begin{subfigure}[t]{.2\textwidth}
    \includegraphics[page=1,width=0.95\textwidth]{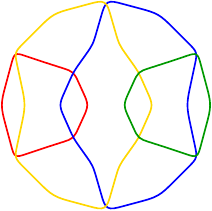}
    \caption{}
    \label{fig:all4_n}  
  \end{subfigure}
  \hfill
  \begin{subfigure}[t]{.2\textwidth}
    \includegraphics[page=1,width=0.95\textwidth]{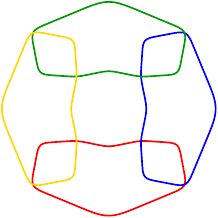}
    \caption{}
    \label{fig:all4_o}
  \end{subfigure}
  \hfill
  \begin{subfigure}[t]{.2\textwidth}
    \includegraphics[angle=-90,origin=c,page=1,width=0.95\textwidth]{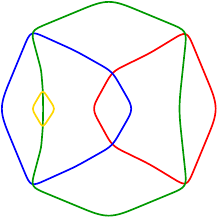}
    \caption{}
    \label{fig:all4_p}  
  \end{subfigure}

  \begin{subfigure}[t]{.2\textwidth}
    \includegraphics[angle=-90,origin=c,page=1,width=0.95\textwidth]{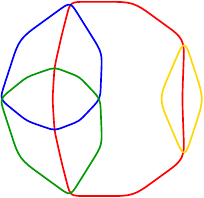}
    \caption{}
    \label{fig:all4_q}
  \end{subfigure}
  \hfill
  \begin{subfigure}[t]{.2\textwidth}
    \includegraphics[angle=-90,origin=c,page=1,width=0.95\textwidth]{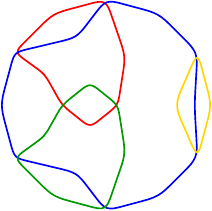}
    \caption{}
    \label{fig:all4_r}  
  \end{subfigure}
  \hfill
  \begin{subfigure}[t]{.2\textwidth}
    \includegraphics[angle=90,origin=c,page=1,width=0.95\textwidth]{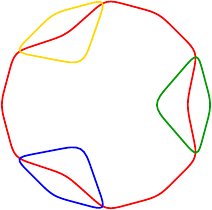}
    \caption{}
    \label{fig:all4_s}  
  \end{subfigure}
  \hfill
  \begin{subfigure}[t]{.2\textwidth}
    \includegraphics[angle=-90,origin=c,page=1,width=0.95\textwidth]{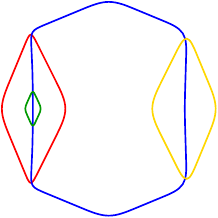}
    \caption{}
    \label{fig:all4_t}  
  \end{subfigure}

  \begin{subfigure}[t]{.2\textwidth}
    \includegraphics[page=1,width=0.95\textwidth]{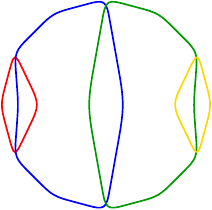}
    \caption{}
    \label{fig:all4_u}
  \end{subfigure}

  \caption{The 21 connected arrangements of $n=4$ pseudocircles.  The
    8 first arrangements (a)--(h) are intersecting.  The arrangements (a),
    (b), and (m) are digon-free. The arrangement (s) is the unique non-cylindrical.
  }
  
  \label{fig:all4}
\end{figure}
%%%%%%%%%%%%%%%%%%%%%%%%%%%%%%%%%%%%%%%%%%%%%%%%%%%%%%%%%%%%%%%%

%%%%%%%%%%%%%%%%%%%%%%%%%%%%%%%%%%%%%%%%%%%%%%%%%%%%%%%%%%%%%%%%%%%

Every triple of great-circles on the sphere induces a Krupp
arrangement, hence, we call an arrangement of pseudocircles 
an \emph{arrangement of great-pseudo\-circles} if every
subarrangement induced by three pseudocircles is a Krupp.

Some authors think of  arrangements of great-pseudocircles
when they speak about arrangements of pseudocircles, this is
e.g.\ common practice in the theory of oriented matroids.
In fact, arrangements of great-pseudocircles serve to
represent rank 3 oriented matroids, cf.~\cite{blwsz-om-93}.
Planar partial cubes can be characterized as the duals
of so-called `non-separating' arrangements of pseudocircles,
these are certain arrangements such that no triple forms a NonKrupp~\cite{ak-cpchn-16}.

\begin{definition*}
An arrangement of pseudocircles is \emph{circularizable} if there is an
isomorphic arrangement of circles.
\end{definition*}

Preceeding our work there have been only few results about circularizability
of arrangements of pseudocircles. Edelsbrunner and Ramos~\cite{EdelsbrunnerRamos1997} 
presented an intersecting arrangement of 6 pseudocircles (with digons) 
which has no realization with circles, i.e., it is not circularizable
(see Figure~\ref{fig:edelsb-ramos}(a)). 
Linhart and Ortner~\cite{LO05} found a non-circularizable
non-intersecting arrangement of 5 pseudocircles with digons 
(see Figure~\ref{fig:n5_nonr_number2}).
They also proved that every intersecting arrangement of at most 4
pseudocircles is circularizable. Kang and M\"uller~\cite{KM14} extended the
result by showing that all arrangements with at most 4 pseudocircles are
circularizable. They also proved that deciding circularizability for connected
arrangements is $\NP$-hard.

%%%%%%%%%%%%%%%%%%%%%%%%%%%%%%%%%%%%%%%%%%%%%%%%%%%%%%%%%%%%%%%%%%%
\section{Overview}
\label{sec:overview-of-results}

In Section~\ref{sec:prelim} we present some background on 
arrangements of pseudocircles and provide tools that will be useful for
non-circularizability proofs. 

\medskip

In Section~\ref{sec:apgc}
we study arrangements of great-pseudocircles 
-- this class of arrangements of pseudocircles is in 
bijection with projective arrangements of pseudolines.
Our main theorem in this section is the Great-Circle Theorem
which allows to transfer knowledge regarding arrangements of
pseudolines to arrangements of pseudocircles.

\begin{restatable}[Great-Circle Theorem]{theorem}{PGCT}
\label{thm:PGC_theorem}
An arrangement of great-pseudocircles 
is circularizable (i.e., has a circle representation) 
if and only if it has a great-circle representation. 
\end{restatable}

Subsequent to the theorem, we present several direct consequences such
as the $\exists \mathbb{R}$-completeness of circularizability.
The complexity class $\exists \mathbb{R}$ 
consists of problems, that can be reduced in polynomial time to 
solvability of a system of polynomial inequalities in several variables over the reals,
and lies inbetween $\NP$ and $\PSPACE$.
Further background on $\exists \mathbb{R}$ can be found in~\cite{Matousek2014,SchaeferStefankovi2017}.

\bigskip

In Sections~\ref{sec:non-circ5} and~\ref{sec:non-circ6}, we present
the full classification of circularizable and non-circularizable
arrangements among all connected arrangements of $5$ pseudocircles and
all digon-free intersecting arrangements of $6$ pseudocircles.  With
the aid of computers we generated the complete lists of connected
arrangements of $n\leq 6$ pseudocircles and of intersecting
arrangements of $n\leq 7$ pseudocircles.  
For the class of arrangements of $n$~great-pseudocircles, 
the numbers were already determined for $n\leq 11$ 
\cite{Knuth1992,AichholzerAurenhammerKrasser2001,Krasser2003,AichholzerKrasser2006}. 
The respective numbers are
shown in Table~\ref{table:numbers}
(cf.\ sequences 
\href{http://oeis.org/A288568}{A288568}, 
\href{http://oeis.org/A296406}{A296406}, and
\href{http://oeis.org/A006248}{A006248}
on the OEIS~\cite{oeis}).
Given the complete lists of
arrangements, we used automatized heuristics to find circle
representations. Examples where the heuristics failed had to be
examined by hand.

%%%%%%%%%%%%%%%%%%%%%%%%%%%%%%%%%%%%%%%%%%%%%%%%%%%%%%%%%%%%%%%%%%%

%%%%%%%%%%%%%%%%%%%%%%%%%%%%%%%%%%%%%%%%%%%%%%%%%%%%%%%%%%%%%%%%%%%%%%%%%%%%%%%%
\begin{table}[htb]\centering
\advance\tabcolsep5pt
\def\arraystretch{1.2}
\begin{tabular}{ l|rrrrrrr }
  $n$   	 	&3 	&4 	&5 	&6 		&7	&8	\\
\hline
\hline
\bfseries{connected}	&3	&21	&984	&609 423	&?		\\
+digon-free		&1	&3	&30	&4 509		&?		\\
\hline
con.+cylindrical	&3	&20	&900	&530 530	&?	\\
+digon-free		&1	&3	&30	&4 477	&?	\\
\hline
\hline
% \bfseries{biconnected}		&2	&15	&838 	&583 396	&?	\\
% +digon-free			&1	&3	&29	&4 499		&?	\\
% \hline
% \hline
\bfseries{intersecting}		&2	&8	&278	&145 058	&447 905 202	&?\\
+digon-free			&1	&2	&14	&2 131		&3 012 972	&?\\
\hline
int.+cylindrical		&2	&8	&278	&144 395	&436 634 633	&?\\
+digon-free			&1	&2	&14	&2 131		&3 012 906	&?\\
\hline
\hline
\bfseries{great-p.c.}s		&1	&1	&1	&4	&11	&135	\\
\end{tabular}

\bigskip

\begin{tabular}{ l|rrrrrrr }
  $n$   	 	&9 	&10 	&11 	%&12
  \\
\hline
\hline
\bfseries{great-p.c.}s		&4 382	&312 356	&41 848 591	%&10 320 613 331
\\
\end{tabular}
% \vskip4mm
\caption{
Number of combinatorially different arrangements of $n$ pseudocircles.
}
\label{table:numbers}
% \vskip-4mm
\end{table}
%%%%%%%%%%%%%%%%%%%%%%%%%%%%%%%%%%%%%%%%%%%%%%%%%%%%%%%%%%%%%%%%%%%%%%%%%%%%%%%%

%%%%%%%%%%%%%%%%%%%%%%%%%%%%%%%%%%%%%%%%%%%%%%%%%%%%%%%%%%%%%%%%%%%

Computational issues and algorithmic ideas are deferred until
Section~\ref{sec:computer-part}.  There we also sketch the heuristics
that we have used to produce circle representations for most of the
arrangements.  The encoded lists of arrangements of up to $n=6$
pseudocircles and circle representations are available on our
webpage~\cite{scheucher_website}. 
Section~\ref{sec:computer-part} also contains 
asymptotic results on the number of arrangements of $n$ pseudocircles 
as well as results on their flip-graph.

The list of circle representations at~\cite{scheucher_website}
together with the non-circularizability proofs given in
Section~\ref{sec:non-circ5} yields the following theorem.

\begin{restatable}{theorem}{NonRealFiveConnected}
\label{thm:non-real-5er_connected}
The four isomorphism classes of arrangements $\AAfiveA$, $\AAfiveB$,
$\AAfiveC$, and~$\AAfiveD$ (shown in Figure~\ref{fig:n5_nonr}) are the
only non-circularizable ones among the 984 isomorphism classes of
connected arrangements of $n=5$ pseudocircles.
\end{restatable}

\begin{corollary}\label{thm:non-real-5er_intersecting}
  The isomorphism class of arrangement~$\AAfiveA$ is the unique
  non-circularizable one among the 278 isomorphism classes of
  intersecting arrangements of $n=5$ pseudocircles.
\end{corollary}

We remark that the arrangements $\AAfiveA$, $\AAfiveB$, $\AAfiveC$,
and $\AAfiveD$ have symmetry groups of order 4, 8, 2, and~4,
respectively.  Also, note that none of the four examples is
digon-free.  Non-circularizability of $\AAfiveB$ was previously shown
by Linhart and Ortner~\cite{LO05}.  We give an alternative proof which
also shows the non-circularizability of~$\AAfiveC$.  Jonathan Wild and
Christopher Jones, contributed sequences A250001 and A288567 to the
On-Line Encyclopedia of Integer Sequences (OEIS)~\cite{oeis}.  These
sequences count certain classes of arrangements of circles and
pseudocircles.  Wild and Jones also looked at circularizability and
independently found Theorem~\ref{thm:non-real-5er_connected} (personal
communication).

%%%%%%%%%%%%%%%%%%%%%%%%%%%%%%%%%%%%%%%%%%%%%%%%%%%%%%%%%%%%%%%%%%%
\begin{figure}[htb]
  \centering
    
\hbox{}  
  \hfill
  \begin{subfigure}[t]{.2\textwidth}
    \centering
    \includegraphics[page=1,width=0.95\textwidth]{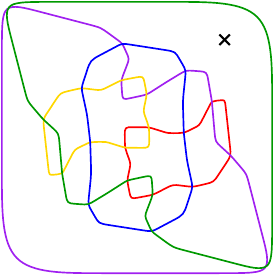}
    \caption{}
    \label{fig:n5_nonr_number1_intersecting}
  \end{subfigure}
  \hfill
  \begin{subfigure}[t]{.2\textwidth}
    \centering
    \includegraphics[page=1,width=0.95\textwidth]{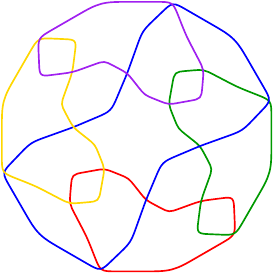}
    \caption{}
    \label{fig:n5_nonr_number2}  
  \end{subfigure}
  \hfill
  \begin{subfigure}[t]{.2\textwidth}
    \centering
    \includegraphics[page=1,width=0.95\textwidth]{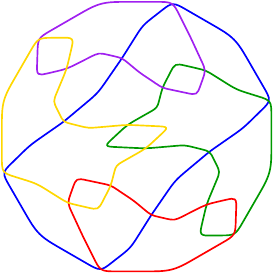}
    \caption{}
    \label{fig:n5_nonr_number3}  
  \end{subfigure}
  \hfill
  \begin{subfigure}[t]{.2\textwidth}
    \centering
    \includegraphics[page=1,width=0.95\textwidth]{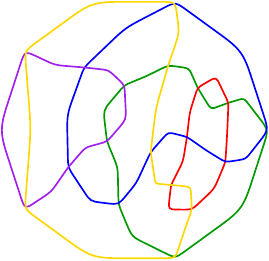}
    \caption{}
    \label{fig:n5_nonr_number4}  
  \end{subfigure}
  \hfill
\hbox{}  
  
  \caption{
  The four non-circularizable arrangements on $n=5$ pseudocircles: 
  (a)~$\AAfiveA$,
  (b)~$\AAfiveB$,
  (c)~$\AAfiveC$, and
  (d)~$\AAfiveD$.
  }
  \label{fig:n5_nonr}
\end{figure}
%%%%%%%%%%%%%%%%%%%%%%%%%%%%%%%%%%%%%%%%%%%%%%%%%%%%%%%%%%%%%%%%%%%

\medskip

%%%%%%%%%%%%%%%%%%%%%%%%%%%%%%%%%%%%%%%%%%%%%%%%%%%%%%%%%%%%%%%%%%%

Concerning arrangements of $6$ pseudocircles,
we were able to fully classify digon-free intersecting arrangements.

\begin{restatable}{theorem}{NonRealSixID}
\label{thm:non-real-6er_intersecting_digonfree}
The three isomorphism classes of arrangements $\AAsixA$, $\AAsixB$, and $\AAsixC$ 
(shown in Figure~\ref{fig:n6_nonr})
are the only non-circularizable
ones among the 2131 isomorphism classes 
of digon-free intersecting arrangements of $n=6$ pseudocircles.
\end{restatable}

%%%%%%%%%%%%%%%%%%%%%%%%%%%%%%%%%%%%%%%%%%%%%%%%%%%%%%%%%%%%%%%%%%%
\begin{figure}[htb]
  \centering
  
\hbox{}  
  \hfill
  \begin{subfigure}[t]{.31\textwidth}
    \centering
    \includegraphics[page=2,width=0.85\textwidth]{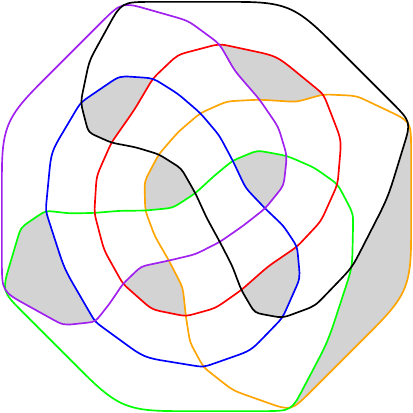}
    \caption{}
    \label{fig:n6_nonr_number1_8triangles}
  \end{subfigure}
  \hfill
  \begin{subfigure}[t]{.31\textwidth}
    \centering
    \includegraphics[page=2,width=0.95\textwidth]{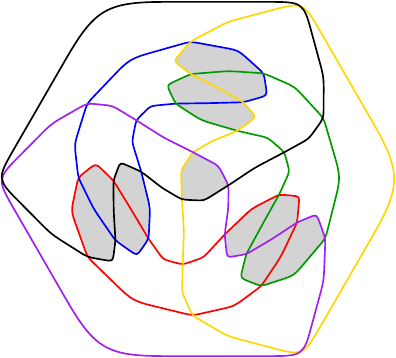}
    \caption{}
    \label{fig:n6_nonr_number2}  
  \end{subfigure}
  \hfill
  \begin{subfigure}[t]{.31\textwidth}
    \centering
    \includegraphics[page=2,width=0.95\textwidth]{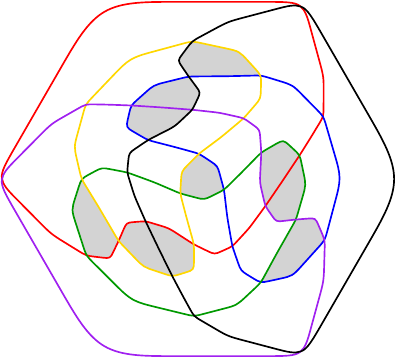}
    \caption{}
    \label{fig:n6_nonr_number3}  
  \end{subfigure}
   \hfill
\hbox{}  
 
  \caption{
  The three non-circularizable digon-free intersecting arrangements for $n=6$:
  (a)~$\AAsixA$,
  (b)~$\AAsixB$, and
  (c)~$\AAsixC$.
  Inner triangles are colored gray.
  Note that in (b) and (c) the outer face is a triangle.
  }
  \label{fig:n6_nonr}
\end{figure}
%%%%%%%%%%%%%%%%%%%%%%%%%%%%%%%%%%%%%%%%%%%%%%%%%%%%%%%%%%%%%%%%%%%

In Section~\ref{sec:non-circ6}, we 
give non-circularizability proofs
for $\AAsixA$, $\AAsixB$, and~$\AAsixC$.
In fact, for the non-circularizability of $\AAsixA$ and $\AAsixB$, respectively,
we have two proofs of different flavors:
One proof (see Section~\ref{sec:non-circ6})
uses continuous deformations 
similar to the proof of the Great-Circle Theorem (Theorem~\ref{thm:PGC_theorem})
and the other proof 
is based on an incidence theorem. The incidence theorem used for $\AAsixA$ 
may be of independent interest:

\begin{restatable}{theorem}{EightPoints}
\label{thm:8points}
  Let $a,b,c,d,A,B,C,D$ be 8 points from $\RR^3$
  such that no five of these points lie in a common plane
  and each of the following 5 subsets of 4
  points is coplanar:\\[3pt]
\centerline{$
  \{a,b,A,B\}, \ 
  \{a,c,A,C\}, \ 
  \{a,d,A,D\}, \ 
  \{b,c,B,C\}, \ 
  \text{and }
  \{b,d,B,D\}. 
  $}\\[3pt]
  Then $\{c,d,C,D\}$ is also coplanar.
\end{restatable}

The proof we give 
is based on determinant cancellation, a technique that we learned from
Richter-Gebert, cf.~\cite{R-G11}.
It turned out that
this incidence theorem
is a slight generalization of the \emph{Bundle theorem} 
(see e.g.~\cite{wikipedia_bundle_theorem}),
where the points are typically assumed to lie on the sphere.

An instance of Theorem~\ref{thm:8points} is obtained by
assigning the eight letters appropriately to the corners of a cube.
Each of the six involved sets then corresponds to the four corners covered 
by two opposite edges of the cube.  
The appropriate assignment of the letters can be derived from Figure~\ref{fig:non-real-6er}.

We remark that the arrangements $\AAsixA$, $\AAsixB$, and $\AAsixC$ have
symmetry groups of order 24, 3, and~6, respectively.
Particularly interesting is the arrangement $\AAsixA$
(Figure~\ref{fig:n6_nonr_number1_8triangles}, see also
Figure~\ref{fig:non-real-6er}). This is the unique
intersecting digon-free arrangement of 6 pseudocircles
which attains the minimum $8$ for the number of triangles
(see~\cite{FS17_APC_TD_arXiv_version}).  

\medskip
Even though we could not complete the classification for 
intersecting arrangements of $6$ pseudocircles,
we provide some further nice incidence theorems and 
non-circularizability proofs in Section~\ref{sec:noncirc_misc}. 
One of them (Figure~\ref{fig:edelsb-ramos}(a)) 
is the example of Edelsbrunner and Ramos~\cite{EdelsbrunnerRamos1997}.
Details on the current status are deferred 
to Section~\ref{ssec:discussion_circularizability_6},
where we also present some further results 
and discuss some open problems.

\medskip

It may be worth mentioning that, 
by enumerating and realizing all arrangement of $n \le 4$ pseudocircles,
we have an alternative proof of the Kang and M\"uller result,
that all arrangements of $n \le 4$ pseudocircles are circularizable~\cite{KM14}.

%%%%%%%%%%%%%%%%%%%%%%%%%%%%%%%%%%%%%%%%%%%%%%%%%%%%%%%%%%%%%%%%%%%%%%%
\section{Preliminaries: Basic Properties and Tools}
\label{sec:prelim}

Stereographic projections map circles to circles (if we consider a
line to be a circle containing the point at infinity),
therefore, circularizability on the
sphere and in the plane is the same concept. 
Arrangements of circles can be mapped to isomorphic
arrangements of circles via M\"obius transformations. 
In this context, the sphere is identified with the extended complex plane 
$\mathbb{C}\cup\{\infty\}$. Note that, for $n \ge 3$, the isomorphism
class of an arrangement of $n$ circles is not covered by M\"obius
transformations. Indeed, if $\CC$ is a simple arrangement of circles,
then $\varepsilon$-perturbations of the circles in size and position will
result in an isomorphic arrangement when $\varepsilon$ is chosen small enough. 

Let $\CC$ be an arrangement of circles represented on the sphere. Each
circle of $\CC$ spans a plane in 3-space, hence, we obtain an
arrangement $\EE(\CC)$ of planes in $\mathbb{R}^3$.  In fact, with a 
sphere $S$ we get a bijection between (not necessarily connected)
circle arrangements on $S$ and arrangements of planes with the
property that each plane of the arrangement intersects~$S$.

Consider two circles $C_1$, $C_2$ of a circle arrangement $\CC$ on $S$ and the
corresponding planes $E_1$, $E_2$ of~$\EE(\CC)$. The intersection of $E_1$ and
$E_2$ is either empty (i.e., $E_1$ and $E_2$ are parallel) or a line $\ell$.
The line $\ell$ intersects $S$ if and only
if $C_1$ and $C_2$ intersect, in fact, $\ell \cap S = C_1 \cap C_2$. 

With three pairwise intersecting circles $C_1$, $C_2$, $C_3$ we obtain
three planes $E_1$, $E_2$, $E_3$ intersecting in a vertex $v$ of
$\EE(\CC)$. It is notable that $v$ is in the interior of $S$ if and
only if the three circles form a Krupp in~$\CC$.  We save this
observation for further reference.

\begin{fact}\label{fact:Krupp_is_inside}
Let $\CC$ be an arrangement of circles represented on the sphere.
Three circles $C_1$, $C_2$, $C_3$ of $\CC$ form a Krupp if and only if
the three corresponding planes  $E_1$, $E_2$, $E_3$ 
intersect in a single point in the interior of~$S$.
\end{fact}

For digons of $\CC$, we also have nessessary conditions in terms of $\EE(\CC)$
and~$S$. 

\begin{fact}\label{fact:digon_characterization}
Let $\CC$ be an arrangement of circles represented on the sphere~$S$.
If a pair of intersecting circles $C_1$, $C_2$ in~$\CC$ forms a digon of~$\CC$,
then the line $E_1 \cap E_2$ has no intersection
with any other plane $E_3$ corresponding to a circle $C_3 \not \in \{C_1,C_2\}$
inside of~$S$.
\end{fact}

%%%%%%%%%%%%%%%%%%%%%%%%%%%%%%%%%%%%%%%%%%%%%%%%%%%%%%%%%%%%%%%%%%%%%%%
\subsection{Incidence Theorems}

The smallest non-stretchable arrangements of pseudolines are closely
related to the incidence theorems of Pappos and Desargues.  A
construction already described by Levi~\cite{Levi26} is depicted in
Figure~\ref{fig:pappos+miquel}(a). Pappos's Theorem states that, in a
configuration of 8 lines as shown in the figure in black, the 3 white
points are collinear, i.e., a line containing two of them also
contains the third. Therefore, the arrangement including the red
pseudoline has no corresponding arrangement of straight lines, i.e.,
it is not stretchable.

Miquel's Theorem asserts that, in a configuration of 5 circles as
shown in Figure~\ref{fig:pappos+miquel}(b) in black, the 4 white
points are cocircular, i.e., a circle containing three of them also
contains the fourth.  Therefore, the arrangement including the red
pseudocircle cannot be circularized.

%%%%%%%%%%%%%%%%%%%%%%%%%%%%%%%%%%%%%%%%%%%%%%%%%%%
%%
% in einem figure environment mit caption
   \calc_figscale{27}
    \begin{figure}[htb]
    \centerline{\input{\path/pappos+miquel.pstex_t}}
    \caption{\label{fig:pappos+miquel}}
    \end{figure}
    VC
{(a)~A non-stretchable arrangement of pseudolines from Pappos's Theorem.\\
(b)~A non-circularizable arrangement of pseudocircles from Miquel's Theorem.}
%%
%%%%%%%%%%%%%%%%%%%%%%%%%%%%%%%%%%%%%%%%%%%%%%%%%%%

Next we state two incidence theorems which will be used in later
proofs of non-circularizability. In the course of the paper we will
meet further incidence theorems such as Lemma~\ref{lem:inzid1},
Lemma~\ref{lem:inzid3}, Theorem~\ref{thm:7kreise},
Lemma~\ref{lem:inzid4}, and
again Miquel's Theorem (Theorem~\ref{thm:Miquels}).

\begin{lemma}[First Four-Circles Incidence Lemma]
\label{lem:inzid2}
Let $\CC$ be an arrangement of four circles $C_1,C_2,C_3,C_4$ such
that none of them is contained in the interior of another one, and 
such that $(C_1,C_2)$, $(C_2,C_3)$, $(C_3,C_4)$, and $(C_4,C_1)$ are
touching. Then there is a circle $C^*$ passing through these four
touching points in the given cyclic order.
\end{lemma} 

We  point the interested reader to the website ``Cut-the-Knot.org''~\cite{Bogomolny},
where this lemma is stated (except for the cyclic order). 
The website also provides an interactive GeoGebra applet,
which nicely illustrates the incidences.

%%%%%%%%%%%%%%%%%%%%%%%%%%%%%%%%%%%%%%%%%%%%%%%%%%
\begin{figure}[htb]
  \centering
  \includegraphics{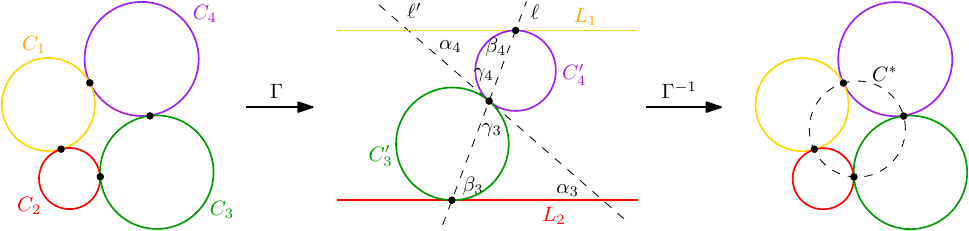}

  \caption{An illustration for the proof of Lemma~\ref{lem:inzid2}.}
  \label{fig:inzid2}  
\end{figure}
%%%%%%%%%%%%%%%%%%%%%%%%%%%%%%%%%%%%%%%%%%%%%%%%%%

\begin{proof}
  Apply a M\"obius transformation $\Gamma$ that maps the touching
  point of $C_1$ and $C_2$ to the point $\infty$ of the extended
  complex plane. This maps $C_1$ and $C_2$ to a pair $L_1,L_2$ of
  parallel lines. The discs of $C_1$ and $C_2$ are mapped to
  disjoint halfplanes. We may assume that $L_1$ and $L_2$ are
  horizontal and that $L_1$ is above $L_2$. Circles $C_3$ and $C_4$ are
  mapped to touching circles $C'_3$ and $C'_4$. Moreover, $C'_3$ is
  touching $L_2$ from above and $C'_4$ touches $L_1$ from
  below. Figure~\ref{fig:inzid2} shows a sketch of the situation.

  Let $\ell'$ be the line, which is tangent to $C'_3$ and $C'_4$ at
  their touching point $p$. Consider the two segments from $p$ to
  $C'_3 \cap L_2$ and from $p$ to $C'_4 \cap L_1$.  
  Elementary considerations show the following equalities of angles: 
  $\alpha_3 = \alpha_4$, $\beta_3 = \gamma_3$, $\beta_4 = \gamma_4$, 
  and $\gamma_3 = \gamma_4$ (cf.\ Figure~\ref{fig:inzid2}). 
  Hence, there is a line $\ell$
  containing the images of the four touching points.
  Consequently, the circle $C^*=\Gamma^{-1}(\ell)$ contains the four
  touching points of~$\CC$, i.e., they are cocircular.
\end{proof} 

The following theorem (illustrated in
Figure~\ref{fig:n6_nonr_number2_proof}) is mentioned by
Richter-Gebert~\cite[page 26]{R-G11} as a relative of Pappos's and
Miquel's Theorem.

\begin{theorem}[\cite{R-G11}]
\label{thm:Incid-3C3L}
Let $C_1,C_2,C_3$ be three circles in the plane such that each pair of them
intersects in two points, and let $\ell_i$ be the line spanned by the two
points of intersection of $C_j$ and $C_k$, for $\{i,j,k\} = \{1,2,3\}$. 
Then $\ell_1$, $\ell_2$, and $\ell_3$ meet in a common point.
\end{theorem} 

%%%%%%%%%%%%%%%%%%%%%%%%%%%%%%%%%%%%%%%%%%%%%%%%%%%%%%%%%%%%%%%%%%%%%%%
\begin{figure}[htb]
  \centering
    \includegraphics{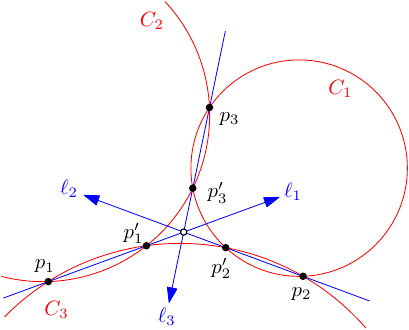}
  \caption{
  An illustration of~Theorem~\ref{thm:Incid-3C3L}.
  }
  \label{fig:n6_nonr_number2_proof}
\end{figure}
%%%%%%%%%%%%%%%%%%%%%%%%%%%%%%%%%%%%%%%%%%%%%%%%%%%%%%%%%%%%%%%%%%%%%%%

\begin{proof}
  Use a stereographic projection $\phi$ to
  map the three circles to circles $C'_1,C'_2,C'_3$ on a sphere~$S$. Consider
  the planes $E'_1,E'_2,E'_3$ spanned by $C'_1,C'_2,C'_3$.  Let $\ell'_i$ be
  the line $E'_j\cap E'_k$, for $\{i,j,k\} = \{1,2,3\}$. 
  Since the arrangement is simple and intersecting,
  the lines $\ell_1',\ell_2',\ell_3'$ are distinct and 
  the intersection $E'_1\cap E'_2 \cap E'_3$ is a single projective point $p$,
  which is contained in each of $\ell'_1,\ell'_2,\ell'_3$.
  The inverse of $\phi$ can be interpreted as a
  central projection from 3-space to the plane.  In this interpretation of
  $\phi^{-1}$, the lines $\ell'_1,\ell'_2,\ell'_3$ are mapped to
  $\ell_1,\ell_2,\ell_3$ and $p$ is mapped to a projective point, i.e., either
$p$ is a point or the lines are parallel.
\end{proof}

%%%%%%%%%%%%%%%%%%%%%%%%%%%%%%%%%%%%%%%%%%%%%%%%%%%%%%%%%%%%

\subsection{Flips and Deformations of Pseudocircles}
\label{ssec:HS-flips}

Let $\CC$ be an arrangement of circles.  Imagine that the circles of
$\CC$ start moving independently, i.e., the position of their centers
and their radii depend on a time parameter $t$ in a continuous
way. This yields a family $\CC(t)$ of arrangements with $\CC(0) =
\CC$. Let us assume that the set~$T$ of all $t$ for which $\CC(t)$ is
not simple or contains touching circles is discrete and for each $t\in
T$ the arrangement~$\CC(t)$ contains either a single point where 3 circles
intersect or a single touching. If $t_1<t_2$ are consecutive in
$T$, then all arrangements $\CC(t)$ with $t \in (t_1,t_2)$ are
isomorphic. Selecting one representative from each such class,
we get a list $\CC_0,\CC_1,\ldots$ of simple arrangements 
such that two consecutive (non-isomorphic) arrangements $\CC_i,\CC_{i+1}$  
are either related by a triangle flip or by a digon flip,
see Figure~\ref{fig:the2flips}.

%%%%%%%%%%%%%%%%%%%%%%%%%%%%%%%%%%%%%%%%%%%%%%%%%%%
%%
% in einem figure environment mit caption
   \calc_figscale{16}
    \begin{figure}[htb]
    \centerline{\input{\path/the2flips.pstex_t}}
    \caption{\label{fig:the2flips}}
    \end{figure}
    VC
{An illustration of the flip operations.}
%%
%%%%%%%%%%%%%%%%%%%%%%%%%%%%%%%%%%%%%%%%%%%%%%%%%%%

We will make use of controlled changes in circle arrangements,
in particular, we grow or shrink specified circles of an arrangement
to produce touchings or points where 3 circles intersect. 
The following lemma will be of use frequently.

\goodbreak

\begin{lemma}[Digon Collapse Lemma]
\label{lem:shrinking_lemma}
Let $\CC$ be an arrangement of circles in the plane 
and let $C$ be one of the circles of~$\CC$,
which intersects at least two other circles from~$\CC$ and
does not fully contain any other circle from~$\CC$ in its interior. 
If $C$ has no incident triangle in its interior, then we can
shrink $C$ into
its interior such that the combinatorics of the arrangement
remain the same except that two digons collapse to touchings. 
Moreover, the two corresponding circles touch $C$ from the outside. 
\end{lemma}

%%%%%%%%%%%%%%%%%%%%%%%%%%%%%%%%%%%%%%%%%%%%%%%%%%%
\begin{figure}[htb]
  \centering
  
  \hbox{}\hfill
    \begin{subfigure}[t]{.25\textwidth}
    \centering
    \includegraphics[page=1]{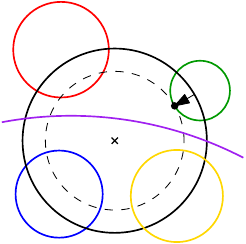}
    %\caption{}
    \label{fig:shrinking_lemma_illustration_page1}  
  \end{subfigure}
  \hfill
  \begin{subfigure}[t]{.25\textwidth}
    \centering
    \includegraphics[page=2]{shrinking_lemma_illustration}
    %\caption{}
    \label{fig:shrinking_lemma_illustration_page2}  
  \end{subfigure}
  \hfill\hbox{}
  
  \caption{An illustration of the Digon Collapse Lemma.}
  \label{fig:shrinking_lemma_illustration}
\end{figure}
%%%%%%%%%%%%%%%%%%%%%%%%%%%%%%%%%%%%%%%%%%%%%%%%%%%

\begin{proof}
As illustrated on the left hand side of Figure~\ref{fig:shrinking_lemma_illustration},
we shrink the radius of $C$ until the first event occurs.
Since $C$ already intersects all other circles of~$\CC$, 
no new digon can be created.
Moreover, since $C$ has no incident triangles in its interior, 
an interior digon collapses. 
We obtain a point where $C$ touches another circle that lies outside of~$C$.
(Note that several digons might collapse at the same time.)

If $C$ has only one touching point $p$, we shrink the radius 
and simultaneously move the center towards~$p$ 
(cf.\ the right hand side of Figure~\ref{fig:shrinking_lemma_illustration})
such that~$p$ stays a touching until a second digon becomes a touching. 
Again the touching point is with a circle that lies outside of~$C$.
\end{proof}

In the following we will sometimes use the dual version of the lemma,
whose statement is obtained from the Digon Collapse Lemma by changing
interior to exterior and outside to inside.  
The validity of the dual
lemma is seen by applying a M\"obius transformation which exchanges
interior and exterior of~$C$.
 
Triangle flips and digon flips are also central to the work of Snoeyink
and Hershberger~\cite{SH91}. They have shown that an arrangement
$\CC$ of pseudocircles can be swept with a sweepfront~$\gamma$
starting at any pseudocircle $C\in \CC$, i.e., $\gamma_0 = C$.  The
sweep consists of two stages, one for sweeping the interior of~$C$,
the other for sweeping the exterior. At any fixed time $t$ the
sweepfront $\gamma_t$ is a closed curve such that $\CC \cup
\{\gamma_t\}$ is an arrangement of pseudocircles. Moreover, this
arrangement is simple except for a discrete set~$T$ of times
where sweep events happen. The sweep events are triangle flips 
or digon flips involving~$\gamma_t$.

%%%%%%%%%%%%%%%%%%%%%%%%%%%%%%%%%%%%%%%%%%%%%%%%%%%%%%%%%%%%%%%%%%%%%%%
\section{Arrangements of Great-Pseudocircles}
\label{sec:apgc}

Central projections map between arrangements of great-circles on a
sphere $S$ and arrangements of lines on a plane. 
Changes of the plane to which we project 
preserve the isomorphism class of the projective arrangement of lines.
In fact, arrangements of lines in the projective plane 
are in one-to-one correspondence to arrangements of great-circles.

In this section we generalize this concept to arrangements of pseudolines and show
that there is a one-to-one correspondence to arrangements of
great-pseudocircles.  As already mentioned, this correspondence is not new (see
e.g.~\cite{blwsz-om-93}).

A \emph{pseudoline} is a simple closed non-contractible curve in the
projective plane. A \emph{(projective) arrangement of pseudolines} is a
collection of pseudolines such that any two intersect in exactly one
point where they cross.  We can also consider arrangements of
pseudolines in the Euclidean plane by fixing a ``line at infinity'' in
the projective plane -- we call this a \emph{projection}.

A Euclidean arrangement of $n$ pseudolines can be represented by
$x$-monotone pseudolines, a special representation of this kind is the
wiring diagram, see e.g~\cite{FelsnerGoodman2016}. 
As illustrated in Figure~\ref{fig:glue2pgc},
an $x$-monotone representation
can be glued with a mirrored copy of itself 
to form an arrangement of $n$ pseudocircles. 
The resulting arrangement is intersecting and has no NonKrupp subarrangement,
hence, it is an arrangement of great-pseudocircles.

\begin{figure}[htb]
\centering
\includegraphics[page=2]{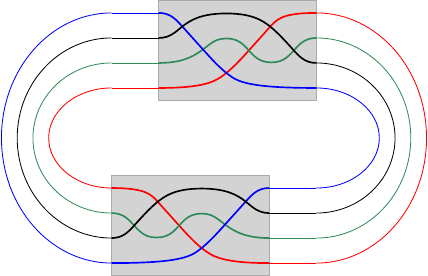}
\caption{
  Obtaining an arrangement of great-pseudocircles from an 
  Euclidean arrangement $\LL$ of pseudolines and its mirrored copy. 
  The gray boxes highlight the arrangement $\LL$
  and its mirrored copy.
}
\label{fig:glue2pgc} 
\end{figure}

For a pseudocircle $C$ of an arrangements of $n$ great-pseudocircles the cyclic
order of crossings on~$C$ is \emph{antipodal}, i.e., the infinite sequence
corresponding to the cyclic order crossings of~$C$ with the other
pseudocircles is periodic of order~$n-1$.  If we consider projections of
projective arrangements of $n$ pseudolines, then this order does not depend on
the choice of the projection.  In fact, projective arrangements of $n$
pseudolines are in bijection with arrangements of $n$ great-pseudocircles.

%%%%%%%%%%%%%%%%%%%%%%%%%%%%%%%%%%%%%%%%%%%%%%%%%%%%%%%%%%%%%%%%%%%%%%%
\subsection{The Great-Circle Theorem and its Applications}
\label{ssec:great_circle_theorem_and_applications}

Let $\AA$ be an arrangement of great-pseudocircles and let
$\LL$ be the corresponding projective arrangement of pseudolines.
Central projections show that, if $\LL$ is realizable with straight
lines, then $\AA$ is realizable with great-circles, and
conversely. 
 
In fact, due to Theorem~\ref{thm:PGC_theorem}, 
it is sufficient that $\AA$ is circularizable
to conclude that $\AA$ is realizable with great-circles 
and $\LL$ is realizable with straight lines. 

%%%%%%%%%%%%%%%%%%%%%%%%%%%%%%%%%%%
\PGCT*

\begin{proof}[Proof of Theorem~\ref{thm:PGC_theorem}]
  Consider an arrangement of circles $\CC$ on the unit sphere $\SS$
  that realizes an arrangement of great-pseudocircles
  as illustrated in Figure~\ref{fig:GCT}(left and center).
  Let $\EE(\CC)$ be the arrangement of planes spanned by the circles of
  $\CC$. Since~$\CC$ realizes an arrangement of great-pseudocircles,
  every triple of circles forms a Krupp, hence, the point of
  intersection of any three planes of $\EE(\CC)$ is in the interior
  of~$\SS$.

  Imagine the planes of $\EE(\CC)$ moving towards the origin.  To be
  precise, for time $t \ge 1$ let
  $\EE_t := \{ \nicefrac{1}{t} \cdot E : E \in \EE(\CC) \}$.  Since
  all intersection points of the initial arrangement
  $\EE_1 = \EE(\CC)$ are in the interior of the unit sphere~$\SS$, the
  circle arrangement obtained by intersecting the moving planes
  $\EE_t$ with~$\SS$ remains the same (isomorphic).   Moreover, 
  as illustrated in Figure~\ref{fig:GCT}(right), 
  every circle in this arrangement converges to a great-circle as
  $t \to +\infty$, and the statement follows. 
\end{proof}

%%%%%%%%%%%%%%%%%%%%%%%%%%%%%%%%%%%%%%%%%%%%
\begin{figure}[htb]
  \centering
    \includegraphics{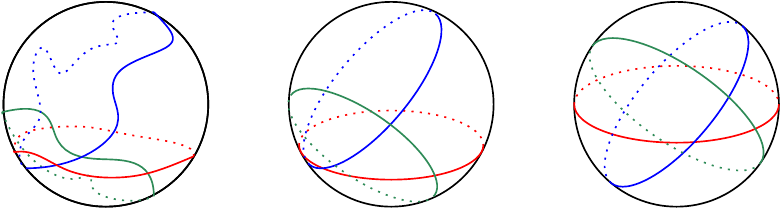}
  \caption{Illustration for the proof of Theorem~\ref{thm:PGC_theorem}.}
  \label{fig:GCT}  
\end{figure}
%%%%%%%%%%%%%%%%%%%%%%%%%%%%%%%%%%%%%%%%%%%%

The Great-Circle Theorem has several interesting consequences.
The following corollary allows the transfer of
results from the world of pseudolines into the world of
(great-)pseudocircles.

\begin{corollary}
An arrangement of pseudolines is stretchable
if and only if the corresponding arrangement of great-pseudocircles
is circularizable.
\end{corollary}

Since deciding stretchability of arrangements of pseudolines is known
to be $\exists\mathbb{R}$-complete (see
e.g.~\cite{Mnev1988,Matousek2014,SchaeferStefankovi2017}), the hardness of stretchability
directly carries over to hardness of circularizability.  To show
containment in $\exists\mathbb{R}$, the circularizability problem has
to be modeled with polynomial inequalities.  This can be done by
taking the centers and radii of the circles as variables and using
polynomial inequalities to prescribe the order of the intersections
along the circles.  

\begin{corollary}
\label{cor:ETR_completeness}
Deciding circularizability is $\exists\mathbb{R}$-complete,
even when the input is restricted to arrangements of great-pseudocircles.
\end{corollary}

It is known that all (not necessarily simple) arrangements of $n \le 8$ pseudolines
are stretchable and that the simple non-Pappos arrangement is the unique
non-stretchable simple projective arrangement of $9$ pseudolines, see
e.g.~\cite{FelsnerGoodman2016}. This again carries over to arrangements of
great-pseudocircles. 

\begin{corollary}
All arrangements of up to $8$ great-pseudocircles are circularizable
and the arrangement
corresponding to the simple non-Pappos arrangement of pseudolines
is the unique non-circularizable arrangement of $9$ great-pseudocircles.
\end{corollary}

Note that the statement of the corollary 
also holds for the non-simple case. A non-simple
arrangement of great-pseudocircles is an arrangement where
three pseudocircles either form a Krupp or the intersection of the three
pseudocircles consists of two points. Gr\"unbaum~\cite{Gr72}
denoted such arrangements as ``symmetric''.
\medskip

Bokowski and Sturmfels~\cite{BokowskiSturmfels1989} have shown 
that infinite families of minimal non-stretchable arrangements of
pseudolines exist, i.e., non-stretchable arrangements
where every proper subarrangement is stretchable. Again, 
this carries over to arrangement of
pseudocircles.

\begin{corollary}
\label{cor:infinite_families_minimal_noncirc}
There are infinite families of 
minimal non-circularizable arrangements of (great-)pseudo\-circles.
\end{corollary}

Mn\"ev's Universality Theorem~\cite{Mnev1988}, see
also~\cite{RG-mutr-95}, has strong implications for pseudoline
arrangements and stretchability.  Besides the hardness of
stretchability, it also shows the existence of arrangements of
pseudolines with a disconnected realization space, that is, there are
isomorphic arrangements of lines such that there is no continuous
transformation which transforms one arrangement into the other within
the isomorphism class. Suvorov~\cite{Suvorov1988} gave an 
explicit example of two such arrangements on $n=13$ lines.

\begin{corollary}
\label{cor:disconnected_realizationspace}
There are circularizable arrangements of (great-)pseudocircles 
with a disconnected realization space.
\end{corollary}

\goodbreak
%%%%%%%%%%%%%%%%%%%%%%%%%%%%%%%%%%%%%%%%%%%%%%%%%%%%%%%%%%%%%%%%%%%%%%%
\section{Arrangements of \texorpdfstring{$5$}{5} Pseudocircles}
\label{sec:non-circ5}

In this section we prove Theorem~\ref{thm:non-real-5er_connected}.
\NonRealFiveConnected*

On the webpage~\cite{scheucher_website} we have the data for circle
realizations of 980 out of the 984 connected arrangements of~$5$
pseudocircles.  The remaining four arrangements in this class are the four
arrangements of Theorem~\ref{thm:non-real-5er_connected}.  Since all
arrangements with $n\leq 4$ pseudocircles have circle representations, there
are no disconnected non-circularizable examples with $n\leq 5$. Hence, the
four arrangements $\AAfiveA$, $\AAfiveB$, $\AAfiveC$, and~$\AAfiveD$ are the
only non-circularizable arrangements with $n\leq 5$. Since $\AAfiveB$,
$\AAfiveC$, and~$\AAfiveD$ are not intersecting,~$\AAfiveA$ is the unique
non-circularizable intersecting arrangement of 5 pseudocircles, this is
Corollary~\ref{thm:non-real-5er_intersecting}.

%%%%%%%%%%%%%%%%%%%%%%%%%%%%%%%%%%%%%%%%%%%%%%%
\subsection{Non-circularizability of~\texorpdfstring{$\AAfiveA$}{N51}}

The arrangement $\AAfiveA$ is depicted in 
Figures~\ref{fig:n5_nonr_number1_intersecting} and~\ref{fig:n5_nonr_number1}.
Since Figure~\ref{fig:n5_nonr_number1_intersecting} 
is ment to illustrate the symmetry of~$\AAfiveA$
while Figure~\ref{fig:n5_nonr_number1} illustrates our non-circularizability proof,
the two drawings of~$\AAfiveA$ differ in the Euclidean plane. However,
we have marked one of the cells of~$\AAfiveA$ by black cross in both figures 
to highlight the isomorphism. 

For the non-circularizability proof of $\AAfiveA$ 
we will make use of the following additional incidence lemma.

\begin{lemma}[Second Four-Circles Incidence Lemma]
\label{lem:inzid1}
Let $\CC$ be an arrangement of four circles $C_1,C_2,\allowbreak{}C_3,C_4$ such that
every pair of them is touching or forms a digon in $\CC$, and every
circle is involved in at least two touchings.
Then there is a circle $C^*$ passing through the digon or touching
point of each of the following pairs of circles
$(C_1,C_2)$, $(C_2,C_3)$, $(C_3,C_4)$, and $(C_4,C_1)$
in this cyclic order.
\end{lemma}

%%%%%%%%%%%%%%%%%%%%%%%%%%%%%%%%%%%%%%%%%%%%
\begin{figure}[htb]
  \centering
    \includegraphics{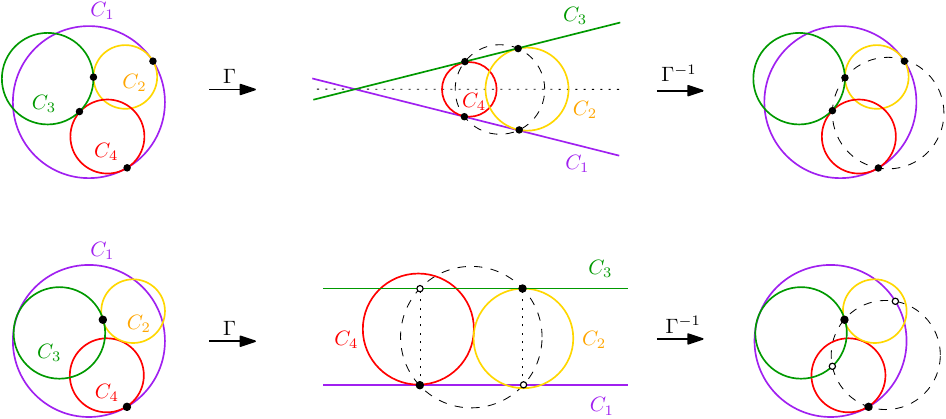}
  \caption{Illustration for the proof of Lemma~\ref{lem:inzid1}.}
  \label{fig:inzid1b}  
\end{figure}
%%%%%%%%%%%%%%%%%%%%%%%%%%%%%%%%%%%%%%%%%%%%

\begin{proof}

  We first deal with the case where $C_1$ and $C_3$ form a digon.  The
  assumptions imply that there is at most one further digon which might then
  be formed by $C_2$ and $C_4$. In particular, the four pairs mentioned in the
  statement of the lemma form touchings and, as illustrated in the first row
  of Figure~\ref{fig:inzid1b}, we will find a circle $C^*$ that is incident to
  those four touching points.  In the following let $p_{i,j}$ denote the
  touching point of $C_i$ and $C_j$.

Think of the circles as being in the extended complex plane.  Apply a M\"obius
transformation $\Gamma$ that maps one of the points of intersection of $C_1$
and $C_3$ to the point $\infty$.  This maps $C_1$ and $C_3$ to a pair of
lines.  The images of $C_2$ and $C_4$ are circles which touch the two lines
corresponding to $C_1$ and $C_3$.
The first row of Figure~\ref{fig:inzid1b} gives an illustration.  Since the
centers of~$C_2$ and~$C_4$ lie on the bisector $\ell$ of the lines
$\Gamma(C_1)$ and $\Gamma(C_3)$, the touchings of $C_2$ and $C_4$ are
symmetric with respect to $\ell$.  Therefore, there is a circle $C$ with
center on $\ell$ that contains the images of the four points $p_{1,2}$,
$p_{2,3}$, $p_{3,4}$, and $p_{4,1}$.  The circle $C^*=\Gamma^{-1}(C)$
contains the four points, i.e., they are cocircular.

Now we consider the case where $C_1$ and $C_3$ touch. 
Again we apply a M\"obius
transformation $\Gamma$ that sends $p_{1,3}$ to $\infty$. 
This maps $C_1$ and $C_3$ to parallel lines, each touched by 
one of $C_2$ and $C_4$. The second
row of Figure~\ref{fig:inzid1b} shows that there is a circle $C$ 
such that  $C^*=\Gamma^{-1}(C)$ has the claimed property.
\end{proof}

\begin{figure}[htb]
  \centering

  \hbox{}\hfill
    \begin{subfigure}[t]{.32\textwidth}
    \centering
    \includegraphics[page=2]{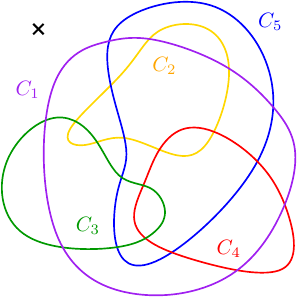}
    %\caption{}
    \label{fig:n5_nonr_number1_page2}  
  \end{subfigure}
  \hfill
  \begin{subfigure}[t]{.35\textwidth}
    \centering
    \includegraphics[page=3]{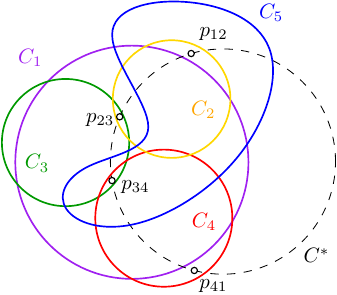}
    %\caption{}
    \label{fig:n5_nonr_number1_page3}  
  \end{subfigure}
  \hfill\hbox{}
  
  \caption{
  An illustration of the non-circularizability proof of~$\AAfiveA$.  
  The auxiliary circle $C^*$ is drawn dashed.
}
  \label{fig:n5_nonr_number1}
\end{figure}

%%%%%%%%%%%%%%%%%%%%%%%%%%%%%%%%%%%%%%%%%%%%%%%%%%%%%

\begin{proof}[Proof (non-circularizability of $\AAfiveA$)]
  Suppose for a contradiction that there is an isomorphic arrangement~$\CC$ of
  circles. We label the circles as illustrated in
  Figure~\ref{fig:n5_nonr_number1}. 
  We apply the Digon Collapse Lemma
  (Lemma~\ref{lem:shrinking_lemma}) to shrink $C_2$, $C_3$, and $C_4$ into
  their respective interiors. We also use the dual of the Digon Collapse Lemma for $C_1$. 
  In the resulting subarrangement $\CC'$ formed by these four transformed circles
  $C_1',C_2',C_3',C_4'$, each of the four circles is involved in at least two touchings. 
  By applying Lemma~\ref{lem:inzid1} to $\CC'$ we obtain a circle $C^*$ which
  passes through the four points~$p_{12}$,~$p_{23}$,~$p_{34}$, and~$p_{41}$ (in this order)
  which respectively are touching points or points from the
  digons of $(C_1',C_2')$, $(C_2',C_3')$, $(C_3',C_4')$, and $(C_4',C_1')$.
  
  Moreover, since the intersection of $C_i'$ and $C_j'$ in~$\CC'$
  is contained in the intersection of $C_i$ and~$C_j$ in~$\CC$, each of the
  four points $p_{12}$, $p_{23}$, $p_{34}$, and $p_{41}$
  lies in the original digon of~$\CC$. 
  It follows that the circle $C_5$ has $p_{12}$ and $p_{34}$ in its interior
  but $p_{23}$ and $p_{41}$ in its exterior (cf.\ Figure~\ref{fig:n5_nonr_number1}).
  By applying Lemma~\ref{lem:inzid1} to $\CC'$ we obtain a circle $C^*$ which
  passes through the points~$p_{12}$,~$p_{23}$,~$p_{34}$, and~$p_{41}$ (in this
  order). Now the two circles $C_5$ and $C^*$ intersect in four
  points.  This is impossible, and hence~$\AAfiveA$ is not circularizable.
\end{proof}

%%%%%%%%%%%%%%%%%%%%%%%%%%%%%%%%%%%%%%%%%%%%%%%%%%%%%

\subsection{Non-circularizability of the connected arrangements 
            \texorpdfstring{$\AAfiveB$}{N52}, \texorpdfstring{$\AAfiveC$}{N53}, and \texorpdfstring{$\AAfiveD$}{N54}}
\label{ssec:non-circ5-connected}

The non-circularizability of $\AAfiveB$ has been shown by Linhart and
Ortner~\cite{LO05}. We give an alternative proof which also shows
the non-circularizability of $\AAfiveC$. The two
arrangements $\AAfiveB$ and $\AAfiveC$ are 
depicted in Figures~\ref{fig:n5_nonr_number2} and~\ref{fig:n5_nonr_number3},
and also in Figures~\ref{fig:n5_nonr_2of3_page2} and~\ref{fig:n5_nonr_3of3_page2}.

% %%%%%%%%%%%%%%%%%%%%%%%%%%%%%%%%%%%%%%%%%%%%%%%%%%%
\begin{figure}[htb]
  \centering
  
  \hbox{}\hfill
    \begin{subfigure}[t]{.4\textwidth}
    \centering
    \includegraphics[page=2]{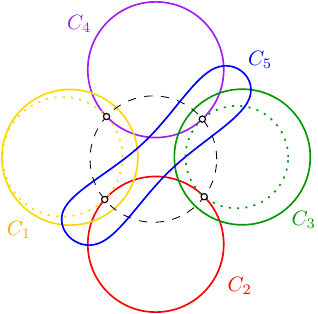}
    \caption{}
    \label{fig:n5_nonr_2of3_page2}  
  \end{subfigure}
  \hfill
  \begin{subfigure}[t]{.4\textwidth}
    \centering
    \includegraphics[page=2]{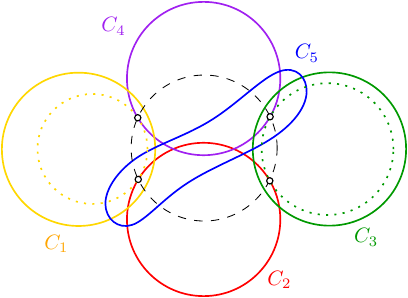}
    \caption{}
    \label{fig:n5_nonr_3of3_page2}  
  \end{subfigure}
  \hfill\hbox{}
  
  \caption{
  An illustration of the non-circularizability proofs of 
  (a)~$\AAfiveB$ and
  (b)~$\AAfiveC$. 
  The auxiliary circle $C^*$ is drawn dashed.
  }
    \label{fig:n5_nonr_2and3}  
\end{figure}
% %%%%%%%%%%%%%%%%%%%%%%%%%%%%%%%%%%%%%%%%%%%%%%%%%%%

\begin{proof}[Proof (non-circularizability of $\AAfiveB$ and $\AAfiveC$)]
  Suppose for a contradiction that there is an isomorphic arrangement $\CC$ of
  circles. We label the circles as illustrated in
  Figure~\ref{fig:n5_nonr_2and3}.  
  Since the respective interiors of $C_1$ and $C_3$ are disjoint,
  we can
  apply the Digon Collapse Lemma
  (Lemma~\ref{lem:shrinking_lemma}) to $C_1$ and~$C_3$.
  This yields an arrangement $\CC'$ with four touching points
  $p_{12},p_{23},p_{34},p_{41}$, where $p_{ij}$ is the touching point of $C_i'$ and~$C_j'$. 
  
  From Lemma~\ref{lem:inzid2} it follows that there is a circle $C^*$
  which passes trough the points $p_{12}$, $p_{23}$, $p_{34}$, and $p_{41}$ in
  this cyclic order. Since the point $p_{ij}$ lies inside the digon formed
  by $C_i$ and $C_j$ in the arrangement~$\CC$, it follows that the circle~$C_5$
  has $p_{12},p_{34}$ in its interior and $p_{23},p_{41}$ in its exterior.
  Therefore, the two circles $C_5$ and~$C^*$ intersect in four points. 
  This is impossible and, therefore, $\AAfiveB$ and $\AAfiveC$
  are not circularizable.
\end{proof}

It remains to prove that $\AAfiveD$ 
(shown in Figures~\ref{fig:n5_nonr_number4} and~\ref{fig:non-real-5er-4}) 
is not circularizable.
In the proof 
we make use of the following incidence lemma.

\begin{lemma}[Third Four-Circles Incidence Lemma]
\label{lem:inzid3}
Let $\CC$ be an arrangement of four circles $C_1,C_2,\allowbreak{}C_3,C_4$ such that
$(C_1,C_2)$, $(C_2,C_3)$, $(C_4,C_1)$, and $(C_3,C_4)$ are touching, 
moreover, $C_4$ is in the interior of $C_1$ and the exterior of $C_3$, and
$C_2$ is in the interior of $C_3$ and the exterior 
of~$C_1$, see Figure~\ref{fig:inzid3-new}.  
Then there is a circle $C^*$ passing through
the four touching points in the given cyclic~order.
\end{lemma} 

%%%%%%%%%%%%%%%%%%%%%%%%%%%%%%%%%%%%%%%%%%%%%%%%%%%%%%%%%%%%%%%
%%
% in einem figure environment mit caption
   \calc_figscale{36}
    \begin{figure}[htb]
    \centerline{\input{\path/inzid3-new.pstex_t}}
    \caption{\label{fig:inzid3-new}}
    \end{figure}
    VC
{An illustration for the proof of Lemma~\ref{lem:inzid3}.}
%%
%%%%%%%%%%%%%%%%%%%%%%%%%%%%%%%%%%%%%%%%%%%%%%%%%%%%%%%%%%%%%%%%

\begin{proof}
  Since $C_1$ touches $C_2$ and $C_4$ which are respectively inside and
  outside $C_3$ the two circles~$C_1$ and~$C_3$ intersect.  Apply a M\"obius
  transformation $\Gamma$ that maps a crossing point of $C_1$ and $C_3$ to the
  point~$\infty$ of the extended complex plane. This maps $C_1$ and $C_3$ to a
  pair $L_1,L_3$ of lines. The images~$C'_2,C'_4$ of $C_2$ and $C_4$ are
  separated by the lines $L_1,L_3$ and each of them touches both
  lines. Figure~\ref{fig:inzid3-new} illustrates the situation.
  The figure also shows that a circle $C'$ through the four touching points exists.
  The circle $C^*=\Gamma^{-1}(C')$ has the claimed properties.
\end{proof} 

% %%%%%%%%%%%%%%%%%%%%%%%%%%%%%%%%%%%%%%%%%%%%%%%%%%%
\begin{figure}[htb]
  \centering
  
    \includegraphics[page=1]{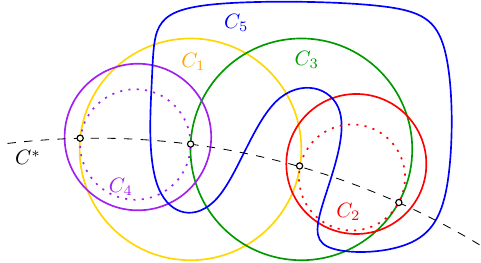}
  
  \caption{
  Illustration of the non-circularizability proof of the arrangement $\AAfiveD$.
  The circle $C^*$ is drawn dashed.
  }
    \label{fig:non-real-5er-4}  
\end{figure}
% %%%%%%%%%%%%%%%%%%%%%%%%%%%%%%%%%%%%%%%%%%%%%%%%%%%

\begin{proof}[Proof (non-circularizability of $\AAfiveD$)]
  Suppose for a contradiction that there is an isomorphic arrangement~$\CC$ of
  circles. We label the circles as illustrated in
  Figure~\ref{fig:non-real-5er-4}.
  We shrink the circles~$C_2$
  and $C_4$ such that each of the pairs 
  $(C_1,C_2)$, $(C_2,C_3)$, $(C_3,C_4)$, and $(C_4,C_1)$ touch. 
  (Note that each of these pairs forms a digon in $\CC$.) 
  With these touchings the four circles
  $C_1,C_2,C_3,C_4$ form the configuration of Lemma~\ref{lem:inzid3}.
  Hence there is a circle~$C^*$ containing the four touching points
  in the given cyclic order. Now the two circles 
  $C^*$ and~$C_5$ intersect in four points. This is impossible and,
  therefore, $\AAfiveD$ is not circularizable.
\end{proof}

%%%%%%%%%%%%%%%%%%%%%%%%%%%%%%%%%%%%%%%%%%%%%%%%%%%%%%%%%%%%%%%%%%%
\section{Intersecting Digon-free Arrangements of \texorpdfstring{$6$}{6} Pseudocircles}
\label{sec:non-circ6}

In this section we prove Theorem~\ref{thm:non-real-6er_intersecting_digonfree}.

%%%%%%%%%%%%%%%%%%%%%%%%%%%%%%%%%%%%%%%%%%%%%%%%
\NonRealSixID*
%%%%%%%%%%%%%%%%%%%%%%%%%%%%%%%%%%%%%%%%%%%%%%%%

We remark that all three arrangements do not have the intersecting
arrangement $\AAfiveA$ as a subarrangement -- otherwise
non-circularizability would follow directly. In fact, $\AAfiveA$ has
no extension to an intersecting digon-free arrangement of six pseudocircles.

On the webpage~\cite{scheucher_website} we have the data of circle realizations
of all 2131 intersecting digon-free arrangements of~$6$ pseudocircles except
for the three arrangements mentioned in
Theorem~\ref{thm:non-real-6er_intersecting_digonfree}.  In the following, we
present two non-circularizability proofs for~$\AAsixA$ and~$\AAsixB$,
respectively, and a non-circularizability proof for~$\AAsixC$.

%%%%%%%%%%%%%%%%%%%%%%%%%%%%%%%%%%%%%%%%%%%%%%%%%%%
\subsection{Non-circularizability of \texorpdfstring{$\AAsixA$}{N61}}
\label{ssec:noncirc_AAsixA}

The arrangement $\AAsixA$ (shown in
Figure~\ref{fig:n6_nonr_number1_8triangles}) is an intersecting
digon-free arrangement.  Our interest in $\AAsixA$ was originally
motivated by our study~\cite{FS17_APC_TD_arXiv_version} of
arrangements of pseudocircles with few triangles. From a computer
search we know that $\AAsixA$ occurs as a subarrangement of every
digon-free arrangement for $n=7,8,9$ with $p_3 < 2n-4$
triangles. Since $\AAsixA$ is not circularizable, neither are these
arrangements. It thus seems plausible that for every arrangement
of $n$ circles $p_3 \geq 2n-4$.  This is the Weak Gr\"unbaum
Conjecture stated in~\cite{FS17_APC_TD_arXiv_version}.

Our first proof is an immediate consequence of the following theorem,
whose proof resembles the proof of the Great-Circle
Theorem (Theorem~\ref{thm:PGC_theorem}).

\begin{theorem}\label{thm:nkt2nc}
  Let $\AA$ be a connected digon-free arrangement of pseudocircles.
  If every triple of pseudocircles which forms 
  a triangle is NonKrupp, then $\AA$ is not circularizable.
\end{theorem}

\begin{proof}
  Assume for a contradiction that there exists an isomorphic
  arrangement of circles $\CC$ on the unit sphere~$\SS$.  Let $\EE(\CC)$
  be the arrangements of planes spanned by the circles of~$\CC$.

  Imagine the planes of $\EE(\CC)$ moving away from the origin.  To be
  precise, for time $t \ge 1$ let $\EE_t := \{ t \cdot E : E \in \EE(\CC) \}$.
  Consider the arrangement induced by intersecting the moving planes $\EE_t$
  with the unit sphere~$\SS$.  Since $\CC$ has NonKrupp triangles, it is not
  a great-circle arrangement and some planes of $\EE(\CC)$ do not contain the
  origin.  All planes from $\EE(\CC)$, which do not contain the origin, will
  eventually lose the intersection with~$\SS$, hence some event has to happen.
  
  When the isomorphism class of the intersection of $\EE_t$ with~$\SS$ changes, 
  we see a triangle flip, or a digon flip,
  or some isolated circle disappears. Since initially there is 
  no digon and no isolated circle, the first event is a triangle flip.
  By assumption, triangles of~$\CC$ correspond to NonKrupp
  subarrangements, hence, the intersection point of their planes
  is outside of~$\SS$ (Fact~\ref{fact:Krupp_is_inside}). This shows that
  a triangle flip event is also impossible.
  This contradiction implies that 
  $\AA$ is non-circularizable.
\end{proof}

\begin{proof}%%
  [Proof (first proof of non-circularizability of $\AAsixA$)]
  The arrangement $\AAsixA$ is intersecting, digon-free, and each of
  the eight triangles of $\AAsixA$ is formed by three
  circles which are a NonKrupp configuration.
  Hence, Theorem~\ref{thm:nkt2nc} implies that $\AAsixA$
  is not circularizable.
\end{proof}

All arrangements known to us whose non-circularizability can be
shown with Theorem~\ref{thm:nkt2nc} contain $\AAsixA$ as a subarrangement 
-- which already shows non-circularizability.
Based on this data we venture the following conjecture: 

\begin{conjecture}
Every connected digon-free arrangement~$\AA$ of pseudocircles with the property, that 
every triple of pseudocircles which forms a triangle in~$\AA$ is NonKrupp, contains $\AAsixA$ as a subarrangement.
\end{conjecture}

Our second proof of non-circularizability of~$\AAsixA$ is based on an
incidence theorem for circles (Theorem~\ref{thm:7kreise}) which is a
consequence of an incidence theorem for points and planes in 3-space
(Theorem~\ref{thm:8points}).  Before going into details, let us describe the
geometry of the arrangement~$\AAsixA$: Consider the non-simple
arrangement~$\AA^\bullet$ obtained from~$\AAsixA$ by contracting each of the
eight triangles into a single point of triple intersection.  The arrangement
$\AA^\bullet$ is circularizable.  A realization is obtained by taking a cube
inscribed in the sphere~$S$ such that each of the eight corners touches the
sphere~$S$.  The arrangement~$\AA^\bullet$ is the intersection of~$S$
with the six planes which are spanned by pairs of diagonally opposite edges of
the cube.

%%%%%%%%%%%%%%%%%%%%%%%%%%%%%%%%%%%%%%%%%%%%%%%%%%%%%%%%%%%%%%%%%%%%
\EightPoints*
%%%%%%%%%%%%%%%%%%%%%%%%%%%%%%%%%%%%%%%%%%%%%%%%%%%%%%%%%%%%%%%%%%%%

\begin{proof}
  By assumption, the five points $a,b,c,d,D$ do not lie in a common plane.
  If we assume -- towards a contradiction -- 
  that the three points $a,b,c$ lie on a common line,
  then $c$ lies in the plane spanned by $\{a,b,A\}$.
  By assumption, the point $B$ also lies in that plane,
  which then contains the five points $\{a,b,A,B,c\}$ 
  -- a contradiction.
  Hence, the three points $a,b,c$ are not collinear.
  
  Now an easy case distinction shows that 
  at most one of the 4-element subsets of the five points $a,b,c,d,D$ is coplanar.
  Since the roles of $d$ and $D$ can be exchanged,
  we assume without loss of generality 
  that $a,b,c,d$ are affinely independent.
  
We now embed $\RR^3$ as the hyperplane $\sum x_i = 1$ into $\RR^4$
such that the four points $a,b,c,d$ become the elements of the standard basis,
namely, $a=\mathbf{e}_1$, $b=\mathbf{e}_2$, $c=\mathbf{e}_3$, and $d=\mathbf{e}_4$. 

Now coplanarity of 4 points can be tested by
evaluating the determinant. Coplanarity of $\{a,b,A,B\}$ 
yields $\det[abAB]=\det(\mathbf{e}_1,\mathbf{e}_2,A,B)=0$.
On the basis of the 5 collinear sets, 
we get the following determinants
and equations:
\[
{\det[abAB] = 0 \atop A_3B_4 = A_4B_3}\qquad
{\det[caCA] = 0 \atop A_4C_2 = A_2C_4}\qquad
{\det[adAD] = 0 \atop A_2D_3 = A_3D_2}\qquad
{\det[cbCB] = 0 \atop B_1C_4 = B_4C_1}\qquad
{\det[bdBD] = 0 \atop B_3D_1 = B_1D_3}\ { \atop . }
\]
Take the product of the left sides of the five equations and the product 
of the right sides. These products are the same. 
We can cancel as much as
possible from the resulting equations and obtain $C_2D_1 = C_1D_2$.
This cancellation works because all relevant factors are non-zero,
which is shown below.
We then obtain that $\det(\mathbf{e}_3,\mathbf{e}_4,C,D)=0$, i.e., the
coplanarity of $\{c,d,C,D\}$.

To show that $A_3 \neq 0$,
suppose that $A_3=0$.
This implies that $\{a,b,d,A\}$ is coplanar.
Suppose that $a,b,A$ are collinear.
Since $\{a,b,d,A\}$ and $\{a,d,A,D\}$ are coplanar, 
the five points $\{a,b,d,A,D\}$ are coplanar 
-- which is a contradiction.
Hence $a,b,A$ are not collinear.
Now, since $\{a,b,A,d\}$ and $\{a,b,A,B\}$ are coplanar,
the five points $\{a,b,A,B,d\}$ are coplanar
-- a contradiction to the assumption $A_3=0$.
The proof that the other relevant factors are non-zero is analogous.
\end{proof}

The theorem implies the following incidence theorem for circles.
\begin{theorem}
\label{thm:7kreise}
Let $C_1,C_2,C_3,C_4$ be four circles and 
let $a,b,c,d,w,x,y,z$ be eight distinct points in $\RR^2$
such that $C_1\cap C_2 = \{a,w\}$,
$C_3\cap C_4 = \{b,x\}$, $C_1\cap C_3 = \{c,y\}$, and
$C_2\cap C_4 = \{d,z\}$. If there is a circle~$C$ containing
$a,b,w,x$, then there is a circle $C'$ containing $c,d,y,z$. 
Moreover, if one of the triples of $C_1,C_2,C_3,C_4$ forms a Krupp, then
$c,d,y,z$ represents the circular order on $C'$.
\end{theorem}
\begin{proof}
  Consider the arrangement of circles on the sphere. The idea is to
  apply Theorem~\ref{thm:8points}. The coplanarity of the~5 sets
  follows because the respective 4 points belong to
  $C,C_1,C_2,C_3,C_4$ in this order. 
  The 5 points $a,b,c,d,z$ do not lie in a common plane
  because otherwise we would have $C_2=C_4$
  (a contradiction  to $C_2\cap C_4 = \{d,z\}$). 
  Analogous arguments show that also
  none of the other 5-element subset of $\{a,b,c,d,x,y,z,w\}$ is coplanar. 
  This shows that Theorem~\ref{thm:8points} can be applied. 
  Regarding the circular order on~$C'$, suppose that
  $C_1,C_2,C_3$ is a Krupp. This implies that $C_2$ separates $c$
  and~$y$. Since $C' \cap C_2 = \{d,z\}$ the points of $\{c,y\}$ and
  $\{d,z\}$ alternate on~$C'$, this implies the claim.
\end{proof}

It is worth mentioning, that 
the second part of the theorem can be strengthened: If one of the
triples of $C_1,C_2,C_3,C_4$ forms a Krupp, then the arrangement
together with $C$ and $C'$ is isomorphic to the simplicial arrangement
$\AA^\bullet$ obtained from $\AAsixA$ by contracting the eight
triangles into triple intersections. 
A \emph{simplicial arrangement} is a non-simple arrangement where all cells are triangles.
The arrangement $\AA^\bullet$ can
be extended to larger simplicial arrangements by adding any subset of
the three circles $C_1^*$, $C_2^*$, $C_3^*$ which are defined as
follows: $C_1^*$ is the circle through the four points
$(C_1\cap C_4)\cup(C_2\cap C_3)$; $C_2^*$ is the circle through the
four points $(C_1\cap C_4)\cup(C\cap C')$; $C_3^*$ is the circle
through the four points $(C_2\cap C_3)\cup(C\cap C')$. In each case
the cocircularity of the four points defining $C_i^*$ is a consequence
of the theorem.  \smallskip

The following lemma is similar to the Digon Collapse Lemma 
(Lemma~\ref{lem:shrinking_lemma}).
By changing interior to exterior and outside to inside 
(by applying a M\"obius transformation),
we obtain a dual version also for this lemma.

\begin{lemma}[Triangle Collapse Lemma]
\label{lem:sl2}
Let $\CC$ be an arrangement of circles in the plane 
and let $C$ be one of the circles of~$\CC$,
which intersects at least three other circles from~$\CC$ and
does not fully contain any other circle from~$\CC$ in its interior. 
If $C$ has no incident digon in its
interior, then we can continuously transform $C$ such that the
combinatorics of the arrangement remains except that two triangles
collapse to points of triple intersection. Moreover, it is possible to
prevent a fixed triangle $T$ incident to $C$ from being the first one
to collapse.
\end{lemma}

\begin{proof}
  Shrink the radius of $C$ until the
  first flip occurs, this must be a triangle flip, i.e., a triangle is
  reduced to a point of triple intersection.  If $C$ has a point $p$
  of triple intersection, shrink $C$ towards $p$, i.e., shrink the
  radius and simultaneously move the center towards~$p$ such that~$p$
  stays incident to $C$. With the next flip a second triangle collapses.

  For the extension let $q \in T\cap C$ be a point. Start the
  shrinking process by shrinking $C$ towards $q$. This prevents $T$
  from collapsing.
\end{proof}

%%%%%%%%%%%%%%%%%%%%%%%%%%%%%%%%%%%%%%%%%%%%%%%%%%%
%%
% in einem figure environment mit caption
   \calc_figscale{30}
    \begin{figure}[htb]
    \centerline{\input{\path/non-real-6er.pstex_t}}
    \caption{\label{fig:non-real-6er}}
    \end{figure}
    VC
{The arrangement $\AAsixA$ with a labeling of its eight triangles.}
%%
%%%%%%%%%%%%%%%%%%%%%%%%%%%%%%%%%%%%%%%%%%%%%%%%%%%

\begin{proof}%%
[Proof (second proof of non-circularizability of $\AAsixA$)]
Suppose for a contradiction that~$\AAsixA$ has a realization~$\CC$.
Each circle of $\CC$ has exactly two incident triangles in the inside
and exactly two on the outside. Apply Lemma~\ref{lem:sl2}
to $C_5$ and to $C_3$ (we refer to the circles with the colors and
labels used in Figure~\ref{fig:non-real-6er}). This collapses
the triangles labeled $a,b,x,w$, i.e., all the triangles 
incident to~$C_5$. Now the green circle~$C_1$, 
the magenta circle~$C_2$, the black circle~$C_3$, the blue circle~$C_4$,
and the red circle~$C_5$ 
are precisely the configuration 
of Theorem~\ref{thm:7kreise} with $C_5$ in the role of~$C$. 
The theorem implies that there is a 
circle $C'$ containing the green-black crossing at $c$, the 
blue-magenta crossing at $d$, the green-black crossing at $y$, and the 
blue-magenta crossing at $z$ in this order. Each consecutive
pair of these crossings is on different sides of the yellow circle~$C_6$,
hence, there are at least four crossings between $C'$ and~$C_6$. 
This is impossible for circles, whence, there
is no circle arrangement $\CC$ realizing $\AAsixA$.
\end{proof}

%%%%%%%%%%%%%%%%%%%%%%%%%%%%%%%%%%%%%%%%%%%%%%%%%%%%%%%%%%%%%%%
\subsection{Non-circularizability of \texorpdfstring{$\AAsixB$}{N62}}

The arrangement $\AAsixB$ is shown in Figure~\ref{fig:n6_nonr_number2} and
Figure~\ref{fig:non-real-6er-Cb}(a).  
We give two proofs for the non-circularizability of $\AAsixB$.
The first one is an immediate consequence of the following theorem,
which -- in the same flavor as Theorem~\ref{thm:nkt2nc} -- 
can be obtained similarly as the proof of the Great-Circle Theorem (Theorem~\ref{thm:PGC_theorem}).

\begin{theorem}\label{thm:nkt2nc_dual}
  Let $\AA$ be an intersecting arrangement of pseudocircles 
  which is not an arrangement of great-pseudocircles.
  If every triple of pseudocircles which forms 
  a triangle is Krupp, 
  then $\AA$ is not circularizable.
\end{theorem}

We outline the proof: Suppose a realization of $\AA$ exists on the sphere.
Continuously move the planes spanned by the circles towards the origin.  The
induced arrangement will eventually become isomorphic to an arrangement of
great-circles.  Now consider the first event that occurs.  As the planes move
towards the origin, there is no digon collapse.  Since $\AA$ is intersecting, no
digon is created, and, since all triangles are Krupp, the corresponding
intersection points of their planes is already inside~$S$.  Therefore, no
event can occur -- a contradiction.

\begin{proof}%%
[Proof (first proof of non-circularizability of $\AAsixB$)]
The arrangement $\AAsixB$ is intersecting but not an arrangement of
great-pseudocircles 
($\AAsixB$ contains a NonKrupp)
and each triangle in $\AAsixB$ is Krupp.
Hence, Theorem~\ref{thm:nkt2nc_dual} implies that $\AAsixB$ is not circularizable.
\end{proof}

Besides $\AAsixB$, there is exactly one other arrangement of $6$ pseudocircles
(with digons, see Figure~\ref{fig:n6_nonr_all_triangles_krupp_number2}) where
Theorem~\ref{thm:nkt2nc_dual} implies non-circularizability.  For $n=7$ there
are eight arrangements where the theorem applies; but each of them has
one of the two $n=6$ arrangements as a subarrangement.

% %%%%%%%%%%%%%%%%%%%%%%%%%%%%%%%%%%%%%%%%%%%%%%%%%%%
\begin{figure}[htb]
  \centering
  
    \includegraphics[page=1,width=0.3\textwidth]{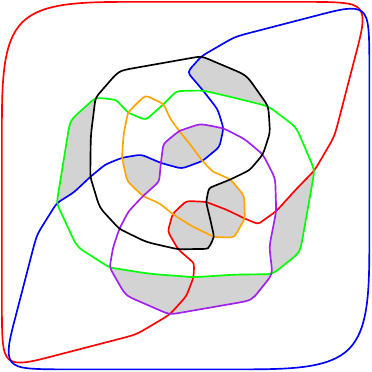}
  
  \caption{
  Another intersecting arrangement of $6$ pseudocircles (with digons) 
  where Theorem~\ref{thm:nkt2nc_dual} applies.
  The arrangement is minimal non-circularizable and has symmetry 2 
  (the colors red-orange, blue-green, and purple-black can be exchanged).
  Triangles are colored gray.
  }
    \label{fig:n6_nonr_all_triangles_krupp_number2}  
\end{figure}
% %%%%%%%%%%%%%%%%%%%%%%%%%%%%%%%%%%%%%%%%%%%%%%%%%%%

\bigskip
Our second proof of non-circularizability of $\AAsixB$ is based
on Theorem~\ref{thm:Incid-3C3L}.

\begin{proof}
[Proof (second proof of non-circularizability of $\AAsixB$)]
Suppose that $\AAsixB$ has a
representation as a circle arrangement~$\CC$. 
We refer to circles
and intersection points via the labels of
the corresponding objects in Figure~\ref{fig:non-real-6er-Cb}(b).

%%%%%%%%%%%%%%%%%%%%%%%%%%%%%%%%%%%%%%%%%%%%%%%%%%%
%%
% in einem figure environment mit caption
   \calc_figscale{30}
    \begin{figure}[htb]
    \centerline{\input{\path/non-real-6er-Cb.pstex_t}}
    \caption{\label{fig:non-real-6er-Cb}}
    \end{figure}
    VC
{(a) The non-circularizable arrangement $\AAsixB$.\\
(b) An illustration for the second proof of non-circularizability of~$\AAsixB$.}
%%
%%%%%%%%%%%%%%%%%%%%%%%%%%%%%%%%%%%%%%%%%%%%%%%%%%%

Let $\ell_i$ be line spanned by $p_i$ and $p_i'$ for $i=1,2,3$.  The directed
line $\ell_1$ intersects $C_{4}$ and $C_{5}$ in the points $p_1,p_1'$ and has
its second intersection with the yellow circle $C_2$ between these points.
After $p_1'$, the line has to cross $C_3$ (magenta), $C_1$ (black), and $C_6$
(red) in this order, i.e., the line behaves as shown in the figure.  Similarly
$\ell_2$ and $\ell_3$ behave as shown.  Let $T$ be the triangle
spanned by the intersection points of the three lines $\ell_1,\ell_2,\ell_3$.
Observe that the gray interior triangle $T'$ of $\CC$ is fully
contained in~$T$.  By applying Theorem~\ref{thm:Incid-3C3L} to the
circles $C_4,C_5,C_6$, we obtain that $\ell_1,\ell_2,\ell_3$ meet in a common
point, and therefore, $T$ and $T'$ are degenerated.  This
contradicts the assumption that $\CC$ is a realization of $\AAsixB$, whence
this arrangement is non-circularizable.
\end{proof}

%%%%%%%%%%%%%%%%%%%%%%%%%%%%%%%%%%%%%%%%%%%%%%%%%%%%%%%%%%%%%%%
\subsection{Non-circularizability of \texorpdfstring{$\AAsixC$}{N63}}

The arrangement  $\AAsixC$ is shown in
Figure~\ref{fig:n6_nonr_number3}
and Figure~\ref{fig:non-real-6er-B}(b).
To prove its non-circularizability, we again use an incidence lemma.
The following lemma is mentioned by Richter-Gebert as a relative of
Pappos's Theorem, cf.~\cite[page 26]{R-G11}.
Figure~\ref{fig:non-real-6er-B}(a) gives an illustration.

\begin{lemma}\label{lem:inzid4}
  Let $\ell_1,\ell_2,\ell_3$ be lines, $C_1',C_2',C_3'$ be circles, and
  $p_{1},p_{2},p_{3},\allowbreak q_{1},q_{2},q_{3}$ be points, such that for
  $\{i,j,k\}=\{1,2,3\}$ point~$p_i$ is incident to line~$\ell_i$,
  circle~$C_j'$, and circle~$C_k'$, while point $q_i$ is incident to
  circle~$C_i'$, line~$\ell_j$, and line~$\ell_k$.  Then $C_1'$, $C_2'$, and
  $C_3'$ have a common point of intersection.
\end{lemma}

%%%%%%%%%%%%%%%%%%%%%%%%%%%%%%%%%%%%%%%%%%%%%%%%%%%
%%
% in einem figure environment mit caption
   \calc_figscale{42}
    \begin{figure}[htb]
    \centerline{\input{\path/non-real-6er-B.pstex_t}}
    \caption{\label{fig:non-real-6er-B}}
    \end{figure}
    VC
{(a) An illustration for Lemma~\ref{lem:inzid4}.\\
(b) The arrangement $\AAsixB$ with 3 dashed pseudolines
illustrating the proof of non-circularizability.
}
%%
%%%%%%%%%%%%%%%%%%%%%%%%%%%%%%%%%%%%%%%%%%%%%%%%%%%

\begin{proof}[Proof (non-circularizability of $\AAsixC$)]
  Suppose that $\AAsixC$ has a representation $\CC$ as a circle arrangement in
  the plane.  We refer to circles and intersection points via the label of the
  corresponding object in Figure~\ref{fig:non-real-6er-B}(b).  As in the
  figure, we assume without loss of generality that the triangular cell
  spanned by $C_4$, $C_5$, and $C_6$ is the outer cell of the arrangement.

Consider the region $R := R_{24} \cup R_{35}$ where 
$R_{ij}$ denotes the intersection of the respective interiors of $C_i$ and~$C_j$.
The two straight-line segments $p_1p'_1$  and $p_3p'_3$ 
are fully contained in~$R_{35}$ and~$R_{24}$, respectively, and
have alternating end points along the boundary of~$R$, 
hence they cross inside the region~$R_{24} \cap R_{35}$. 

From rotational symmetry we obtain that the three straight-line segments 
$p_1p'_1$, $p_2p'_2$, and $p_3p'_3$ intersect pairwise.

For $i =1,2,3$, 
let $\ell_i$ denote the line spanned by $p_i$ and~$p_i'$,
let $q_i$ denote the intersection-point of $\ell_{i+1}$ and $\ell_{i+2}$,
and let $C_i'$ denote the circle spanned by $q_i,p_{i+1},p_{i+2}$ (indices modulo~3).
Note that $\ell_i$ contains $p_i,q_{i+1},q_{i+2}$.
These are precisely the conditions for the incidences of points, lines, and circles in
Lemma~\ref{lem:inzid4}. Hence, the three circles $C_1'$, $C_2'$, and~$C_3'$
intersect in a common point (cf.\ Figure~\ref{fig:non-real-6er-B}(a)). 

Let $T$ be the triangle with corners $p_1,p_2,p_3$. 
Since $p_2$ and $p_3$ are on~$C_1$, and 
$q_1$ lies inside of~$C_1$,
we find that the intersection of the interior of $C_1'$ with~$T$ 
is a subset of the intersection of the interior of $C_1$ with~$T$. 
The respective containments also hold for $C_2'$ and~$C_2$
and for $C_3'$ and~$C_3$.
Moreover, since  $C_1'$, $C_2'$, and $C_3'$ intersect in a common point, the union of the 
interiors of $C_1'$, $C_2'$, and $C_3'$ contains~$T$. 
Hence, the union
of interiors of the $C_1$, $C_2$, and $C_3$ also contains~$T$. 
This shows that in~$\CC$ there is no face corresponding to the gray triangle; 
see Figure~\ref{fig:non-real-6er-B}(b). 
This contradicts the assumption that $\CC$ is a
realization of~$\AAsixC$, whence the arrangement is non-circularizable.
\end{proof}

%%%%%%%%%%%%%%%%%%%%%%%%%%%%%%%%%%%%%%%%%%%%%%%%%%%%%%%%%%%%%%%%%%%%%%%%%%%
\section{Additional Arrangements with \texorpdfstring{$n=6$}{n=6}}
\label{sec:noncirc_misc}

In the previous two sections we have exhibited all non-circularizable
arrangements with $n\leq 5$ and all non-circularizable intersecting
digon-free arrangements with $n=6$. With automatized procedures we
managed to find circle representations of 98\% of the connected
digon-free arrangements and of 90\% of the intersecting arrangements
of 6 pseudocircles. Unfortunately, the numbers of remaining candidates
for non-circularizability are too large to complete the classification
by hand (cf.\ Section~\ref{ssec:discussion_circularizability_6}). 
In this section we show non-circularizability of a few
of the remaining examples which we consider to be interesting.
As a criterion for being interesting we used the order of the symmetry
group of the arrangement. 
The symmetry groups have been determined as the automorphism groups 
of the primal-dual graphs using SageMath~\cite{sagemath_website,sagemath_graphtheory_manual}. 

In Subsection~\ref{ssec:3inters+digon} we show non-circularizability
of the three intersecting arrangements of $n=6$ pseudocircles (with digons) 
depicted in Figure~\ref{fig:misc_int6_sym6}. The symmetry group of these three
arrangements is of order~6. All the remaining examples of intersecting
arrangements with $n=6$, where we do not know about
circularizability, have a symmetry group of order at most~3.

%%%%%%%%%%%%%%%%%%%%%%%%%%%%%%%%%%%%%%%%%%%%%%%%%%%%%%%%%%%%%%%%%%%%%%%%%%%
\begin{figure}[htb]
  \centering
  
  \hbox{}\hfill
    \begin{subfigure}[t]{.3\textwidth}
    \centering
    \includegraphics[page=1,width=0.95\textwidth]{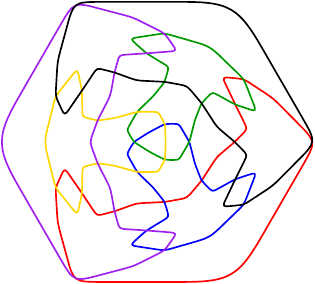}
    \caption{}
    \label{fig:misc_int6_sym6_1_ER}
  \end{subfigure}
  \hfill
    \begin{subfigure}[t]{.3\textwidth}
    \centering
    \includegraphics[page=1,width=0.95\textwidth]{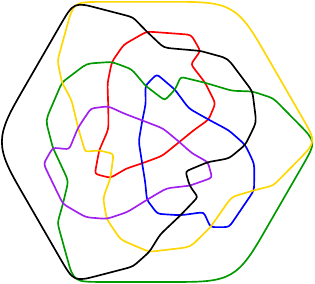}
    \caption{}
    \label{fig:misc_int6_sym6_2}
  \end{subfigure}
  \hfill
    \begin{subfigure}[t]{.3\textwidth}
    \centering
    \includegraphics[page=1,width=0.95\textwidth]{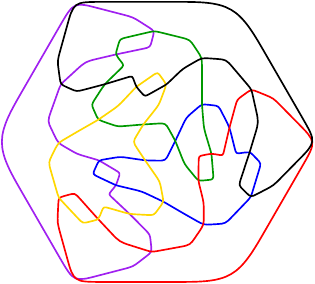}
    \caption{}
    \label{fig:misc_int6_sym6_3}
  \end{subfigure}
  \hfill\hbox{}
  
  \caption{Three intersecting arrangements of $n=6$ pseudocircles with symmetry 6.\\
  (a)~$\AAsixER$ (b)~$\AAsixMiscIntB$ (c)~$\AAsixMiscIntC$.}
  \label{fig:misc_int6_sym6}
\end{figure}
%%%%%%%%%%%%%%%%%%%%%%%%%%%%%%%%%%%%%%%%%%%%%%%%%%%%%%%%%%%%%%%%%%%%%%%%%%%

In Subsection~\ref{ssec:3conn-digon} we show non-circularizability of
the three connected digon-free arrangements of $6$ pseudocircles
depicted in Figure~\ref{fig:misc_con6_high_sym}. The symmetry group of
these three arrangements is of order~24 or~8. In
Subsection~\ref{ssec:extra2} we show non-circularizability of
two additional connected digon-free arrangements of $6$ pseudocircles.
The examples are shown in Figure~\ref{fig:extra2}, the symmetry group
of these two arrangements is of order~4.
All the remaining examples of
connected digon-free arrangements with $n=6$, where we do not know about
circularizability, have a symmetry group of order 2 or~1.

%%%%%%%%%%%%%%%%%%%%%%%%%%%%%%%%%%%%%%%%%%%%%%%%%%%%%%%%%%%%%%%%%%%%%%%%%%%
\subsection{Non-circularizability of three intersecting arrangements
  with \texorpdfstring{$n=6$}{n=6}}\label{ssec:3inters+digon}

In this subsection we prove non-circularizability of the three
arrangements $\AAsixER$, $\AAsixMiscIntB$, and $\AAsixMiscIntC$
shown in Figure~\ref{fig:misc_int6_sym6}. The non-circularizability
of $\AAsixER$ was already shown by Edelsbrunner and Ramos~\cite{EdelsbrunnerRamos1997},
the name of the arrangement reflects this fact. The other names
are built, such that the subscript of the~$\cal N$ is the number
of pseudocircles, the first part of the superscript indicates that the
arrangement is intersecting with a symmetry group of order~6, and the number after
the colon is the counter. Accordingly, the arrangement~$\AAsixER$ can 
also be denoted as ${{\cal N}_{6}^{i6:1}}$.

%%%%%%%%%%%%%%%%%%%%%%%%%%%%%%%%%%%%%%%%%%%%%%%%%%%%%%%%%%%%%%%%%%%%%%%%%%%
\subsubsection{Non-circularizability of the Edelsbrunner--Ramos example \texorpdfstring{$\AAsixER$}{N6ER}}
\label{sssec:noncirc_edelsbrunner_ramos}

The arrangement $\AAsixER$ is shown in Figure~\ref{fig:misc_int6_sym6_1_ER}.
As in the original proof~\cite{EdelsbrunnerRamos1997} the argument is based on
considerations involving angles.

%%%%%%%%%%%%%%%%%%%%%%%%%%%%%%%%%%%%%%%%%%%%%%%%%%%
%%
% in einem figure environment mit caption
   \calc_figscale{22}
    \begin{figure}[htb]
    \centerline{\input{\path/edelsb-ramos.pstex_t}}
    \caption{\label{fig:edelsb-ramos}}
    \end{figure}
    VC
{(a) The Edelsbrunner--Ramos example $\AAsixER$.\\
 (b) Comparing the angle at $a$ and the corresponding angle of $T$.\\
 (c) Labels for the vertices of the inner subarrangement $\AA_I$.
 }	
%%
%%%%%%%%%%%%%%%%%%%%%%%%%%%%%%%%%%%%%%%%%%%%%%%%%%%

Figure~\ref{fig:edelsb-ramos}(a) shows a representation of the
arrangement $\AAsixER$ consisting of a subarrangement $\AA_O$ formed
by the three \textit{outer pseudocircles} $C_1,C_2,C_3$ and a second
subarrangement $\AA_I$ formed by the three \textit{inner circles}
$C_4,C_5,C_6$.

Suppose that there is a circle representation $\CC$ of $\AAsixER$.
Let $\CC_O$ and $\CC_I$ be the subarrangements of~$\CC$ which
represent $\AA_O$ and $\AA_I$, respectively.  
For each outer circle $C_i$ from $\CC_O$ consider 
a straight-line segment $s_i$ that connects two
points from the two digons which are formed by $C_i$ with inner circles.  
The segment $s_i$ is fully contained in $C_i$.  
Let $\ell_i$ be the line supporting $s_i$ and
let $T$ be the triangle bounded by $\ell_1$, $\ell_2$, and $\ell_3$.

We claim that $T$ contains the inner triangle of $\CC_O$.
Indeed, if three circles form a NonKrupp where the outer face is a
triangle and with each circle we have
a line which intersects the two digons incident to the circle, then
the three lines form a triangle containing the inner triangular cell
of the NonKrupp arrangement.

The inner triangle of $\CC_O$ contains the four inner triangles
of~$\CC_I$.  Let $a,b,c$ be the three crossing points on the outer
face of the subarrangement $\CC_I$.  Comparing the inner angle at $a$,
a crossing of $C_4$ and $C_5$, and the corresponding angle of $T$, 
i.e., the angle formed by $\ell_1$ and $\ell_2$, we
claim that the inner angle at $a$ is smaller. 
To see this let us assume that the common tangent $h$ of $C_4$ and $C_5$
on the side of $a$ is horizontal. Line $\ell_1$ has both crossings 
with $C_5$ above $a$ and also intersects with $C_6$. This implies that the 
slope of $\ell_1$ is positive but smaller than the slope of the tangent
at $C_5$ in $a$. Similarly, the slope of $\ell_2$ is negative but
larger than the slope of the tangent
at $C_4$ in $a$. This is the claim, see Figure~\ref{fig:edelsb-ramos}(b). 

The respective statements hold for the inner angles
at $b$ and $c$, and the corresponding angles of $T$.  Since the
sum of angles of $T$ is $\pi$, we conclude that the sum of the inner
angles at $a$, $b$, and $c$ is less than~$\pi$.

The sum of inner angles at $a$, $b$, $c$ equals the sum of inner
angles at $a'$, $b'$, $c'$, see Figure~\ref{fig:edelsb-ramos}(c). This
sum, however, clearly exceeds $\pi$.  The contradiction shows that
$\AAsixER$ is not circularizable.

%%%%%%%%%%%%%%%%%%%%%%%%%%%%%%%%%%%%%%%%%%%%%%%%%%%%%%%%%%%%%%%%%%%%%%%%%%%
\subsubsection{Non-circularizability of \texorpdfstring{$\AAsixMiscIntB$}{N6i6:2}}

The arrangement $\AAsixMiscIntB$ is shown in
Figure~\ref{fig:misc_int6_sym6_2} and again in
Figure~\ref{fig:Ni62}(a).  This figure also shows some shaded
triangles, three of them are gray and three are pink.  

%%%%%%%%%%%%%%%%%%%%%%%%%%%%%%%%%%%%%%%%%%%%%%%%%%%
%%
% in einem figure environment mit caption
   \calc_figscale{48}
    \begin{figure}[htb]
    \centerline{\input{\path/Ni62.pstex_t}}
    \caption{\label{fig:Ni62}}
    \end{figure}
    VC
{(a) The arrangement $\AAsixMiscIntB$ with some triangle faces emphasized.\\
 (b) After collapsing the shaded triangles. \\
 (c) After moving the point $p_1$ to infinity.
}
%%
%%%%%%%%%%%%%%%%%%%%%%%%%%%%%%%%%%%%%%%%%%%%%%%%%%%

Suppose that
$\AAsixMiscIntB$ has a circle representation~$\CC$.  Each of $C_1$,
$C_2$, and~$C_3$ has two triangles and no digon on its interior
boundary. One of the two triangles is gray the other
pink. Lemma~\ref{lem:sl2} allows to shrink the three circles $C_1$,
$C_2$, $C_3$ of $\CC$ into their respective interiors such that in
each case the shrinking makes a pink triangle collapse.
Let $p_i$ be the point of triple intersection of $C_i$, for $i=1,2,3$.  
Further shrinking $C_i$ towards $p_i$ makes another triangle collapse.  
At this second collapse two triangles disappear, one of them a gray one,
and $C_i$ gets incident to~$p_{i-1}$ (with the understanding that $1-1 =
3$). Having done this for each of the three circles yields a circle
representation for the (non-simple) arrangement shown in 
Figure~\ref{fig:Ni62}(b).

To see that this arrangement has no circle representation apply a M\"obius
transformation that maps the point $p_1$ to the point $\infty$ of the extended
complex plane.  This transforms the four circles $C_1$, $C_2$, $C_4$, $C_5$,
which are incident to $p_1$, into lines.  The two remaining circles $C_3$ and
$C_6$ intersect in $p_2$ and~$p_3$.  The lines of $C_2$ and $C_5$ both have
their second intersections with $C_3$ and $C_6$ separated by $p_2$, hence,
they both avoid the lens formed by $C_3$ and $C_6$. The line of $C_1$ has its
intersections with $C_2$ and $C_5$ in the two components of the gray
double-wedge of $C_2$ and $C_5$, see Figure~\ref{fig:Ni62}(c).  Therefore, the
slope of~$C_1$ belongs to the slopes of the double-wedge. However, the line
of~$C_1$ has its second intersections with~$C_3$ and~$C_6$ on the same side of
$p_3$ and, therefore, it has a slope between the tangents of $C_3$ and $C_6$
at~$p_3$.  These slopes do not belong to the slopes of the gray double-wedge.
This contradiction shows that a circle representation $\CC$ of
$\AAsixMiscIntB$ does not exist.

%%%%%%%%%%%%%%%%%%%%%%%%%%%%%%%%%%%%%%%%%%%%%%%%%%%%%%%%%%%%%%%%%%%%%%%%%%%
\subsubsection{Non-circularizability of \texorpdfstring{$\AAsixMiscIntC$}{N6i6:3}}

The arrangement $\AAsixMiscIntC$ is shown in
Figure~\ref{fig:misc_int6_sym6_3}
and again in Figure~\ref{fig:Ni63}(a).
This figure also shows some shaded triangles, three of them are gray
and three are pink.  

%%%%%%%%%%%%%%%%%%%%%%%%%%%%%%%%%%%%%%%%%%%%%%%%%%%
%%
% in einem figure environment mit caption
   \calc_figscale{42}
    \begin{figure}[htb]
    \centerline{\input{\path/Ni63.pstex_t}}
    \caption{\label{fig:Ni63}}
    \end{figure}
    VC
{(a) The arrangement $\AAsixMiscIntC$ with some triangle faces emphasized.\\
 (b) After collapsing the pink shaded triangles. \\
 (c) A detail of the arrangement after the second phase of collapses.
}	
%%
%%%%%%%%%%%%%%%%%%%%%%%%%%%%%%%%%%%%%%%%%%%%%%%%%%%

Suppose that $\AAsixMiscIntB$ has a circle
representation~$\CC$.  Each of $C_1$, $C_2$ and~$C_3$ has two
triangles and no digon on its exterior boundary. One of the two
triangles is gray the other is pink (we disregard the exterior
triangle because it will not appear in a first flip when expanding a
circle).  The dual form of Lemma~\ref{lem:sl2} allows to expand the
three circles $C_1$, $C_2$, $C_3$ of $\CC$ into their respective
exteriors such that in each case the expansion makes a pink triangle
collapse. 

Figure~\ref{fig:Ni63}(b) shows a pseudocircle
representation of the arrangement after this first phase of
collapses. In the second phase, we modify the circles $C_i$, for
$i=4,5,6$. We explain what happens to $C_5$, the two other circles are
treated similarly with respect to the rotational symmetry.
Consider the circle $C_5'$, which contains $p_1$ and 
shares $p_3$ and the red point with~$C_5$. 
This circle is obtained by
shrinking $C_5$ on one side of the line containing $p_3$ and the red point, and 
by expanding $C_5$ on the other side of the line. It is easily verified
that the collapse of the gray triangle at $p_1$ is the first event in this
process. 

Figure~\ref{fig:Ni63}(c) shows the inner triangle formed by $C_1$,
$C_2$, and~$C_3$ together with parts of $C_4'$, $C_5'$, and~$C_6'$. At
each of the three points, the highlighted red angle is smaller than the
highlighted gray angle. However, the red angle at $p_i$ is formed by
the same two circles as the gray angle at $p_{i+1}$, whence, the two
angles are equal. This yields a contradictory cyclic chain of
inequalities.  The contradiction shows that a circle representation
$\CC$ of $\AAsixMiscIntC$ does not exist.

%%%%%%%%%%%%%%%%%%%%%%%%%%%%%%%%%%%%%%%%%%%%%%%%%%%%%%%%%%%%%%%%%%%%%%%%%%%
\subsection{Non-circularizability of three connected digon-free arrangements
  with \texorpdfstring{$n=6$}{n=6}}\label{ssec:3conn-digon}

In this subsection we prove non-circularizability of the three
arrangements $\AAsixMiscConVierundZwanzig$, $\AAsixMiscConAchtA$, and
$\AAsixMiscConAchtB$ shown in Figure~\ref{fig:misc_con6_high_sym}. 

%%%%%%%%%%%%%%%%%%%%%%%%%%%%%%%%%%%%%%%%%%%%%%%%%%%%%%%%%%%%%%%%%%%%%%%%%%%
\begin{figure}[htb]
  \centering
  
  \hbox{}\hfill
    \begin{subfigure}[t]{.26\textwidth}
    \centering
    \includegraphics[page=1,width=0.95\textwidth]{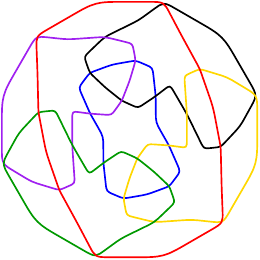}
    \caption{}
    \label{fig:misc_con6_sym24_1}
  \end{subfigure}
  \hfill
    \begin{subfigure}[t]{.26\textwidth}
    \centering
    \includegraphics[page=1,width=0.95\textwidth]{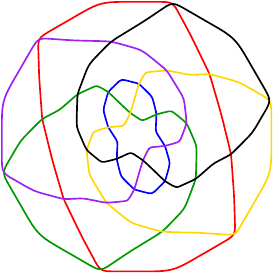}
    \caption{}
    \label{fig:misc_con6_sym8_2}
  \end{subfigure}
  \hfill
    \begin{subfigure}[t]{.26\textwidth}
    \centering
    \includegraphics[page=1,width=0.95\textwidth]{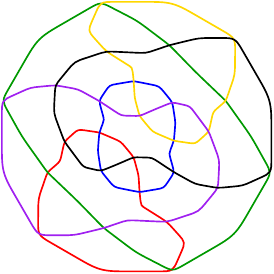}
    \caption{}
    \label{fig:misc_con6_sym8_1}
  \end{subfigure}
  \hfill\hbox{}
  
  \caption{A digon-free connected arrangement of $n=6$ pseudocircles
    with symmetry group of order~24, and two with a symmetry group of order~8:
  (a)~$\AAsixMiscConVierundZwanzig$ (b)~$\AAsixMiscConAchtA$ 
  (c)~$\AAsixMiscConAchtB$.}
  \label{fig:misc_con6_high_sym}
%%%%%%%%%%%%%%%%%%%%%%%%%%%%%%%%%%%%%%%%%%%%%%%%%%%%%%%%%%%%%%%%%%%%%%%%%%%
\end{figure}

%%%%%%%%%%%%%%%%%%%%%%%%%%%%%%%%%%%%%%%%%%%%%%%%%%%%%%%%%%%%%%%%%%%%%%%%%%%
\subsubsection{Non-circularizability of \texorpdfstring{$\AAsixMiscConVierundZwanzig$}{N6c24}
  and \texorpdfstring{$\AAsixMiscConAchtA$}{N6c8:1}}

The proof of non-circularizability of the two arrangements 
is based on Miquel's Theorem. For proofs of 
the theorem we refer to~\cite{R-G11}.

\begin{theorem}[Miquel's Theorem]
\label{thm:Miquels}
Let $C_1,C_2,C_3,C_4$ be four circles
and let $C_1\cap C_2 = \{a,w\}$, $C_2\cap C_3 = \{b,x\}$,
$C_3\cap C_4 = \{c,y\}$,  and $C_4\cap C_1 = \{d,z\}$.
If there is a circle $C$ containing
$a,b,c,d$, then there is a circle $C'$ containing $w,x,y,z$. 
\end{theorem}

%%%%%%%%%%%%%%%%%%%%%%%%%%%%%%%%%%%%%%%%%%%%%%%%%%%
%%
% in einem figure environment mit caption
   \calc_figscale{46}
    \begin{figure}[htb]
    \centerline{\input{\path/N24.pstex_t}}
    \caption{\label{fig:N24}}
    \end{figure}
    VC
{(a) The arrangement $\AAsixMiscConVierundZwanzig$.
 (b) $\AAsixMiscConVierundZwanzig$ after collapsing four triangles.
 (c) The arrangement $\AAsixMiscConAchtA$.}	
%%
%%%%%%%%%%%%%%%%%%%%%%%%%%%%%%%%%%%%%%%%%%%%%%%%%%%

The arrangement $\AAsixMiscConVierundZwanzig$ is shown in
Figure~\ref{fig:misc_con6_sym24_1} and again in
Figure~\ref{fig:N24}(a).  Suppose that $\AAsixMiscConVierundZwanzig$
has a circle representation~$\CC$.  Circle $C_5$ has six triangles in
its exterior. These triangles are all incident to either the crossing
of $C_1$ and $C_4$ or the crossing of $C_2$ and $C_3$. Hence,
by Lemma~\ref{lem:sl2} we can grow~$C_5$ into its exterior
to get two triple intersection points $p_1=C_1\cap C_4 \cap C_5$ and
$p_2=C_2\cap C_3 \cap C_5$. 
The situation in the interior of $C_6$ is
identical to the situation in the exterior of $C_5$. Hence, by
shrinking~$C_6$ we get two additional triple intersection points
$p_3=C_1\cap C_2 \cap C_6$ and $p_4=C_3\cap C_4 \cap C_6$.  This
yields the (non-simple) arrangement shown in Figure~\ref{fig:N24}(b).
Now grow the circles $C_1,C_2,C_3,C_4$ to the outside while keeping
each of them incident to its two points~$p_i$, this makes them shrink
into their inside at the `short arc'.  Upon this growth process, the
gray crossings $p_{12}$, $p_{23}$, $p_{34}$, and $p_{14}$ move away
from the blue circle $C_6$.  Hence the process can be continued
until the upper and the lower triangles collapse, i.e., until $p_{12}$
and $p_{34}$ both are incident to $C_5$. Note that we do not care
about $p_{23}$ and $p_{14}$, they may have passed to the other side of~$C_5$.  
The collapse of the upper and the lower triangle yields two
additional triple intersection points $C_1\cap C_2 \cap C_5$ and
$C_3\cap C_4 \cap C_5$. The circles $C_1,C_2,C_3,C_4$ together with
$C_5$ in the role of $C$ form an instance of Miquel's Theorem
(Theorem~\ref{thm:Miquels}).  
Hence, there is a circle $C'$ traversing the four points $p_3,p_{14},p_4,p_{23}$.
The points $p_3,p_4$ partition $C'$ into two arcs, 
each containing one of the gray points $p_{14}, p_{23}$.
Now $C'$ shares the points $p_3$ and $p_4$ with~$C_6$
while the gray points $p_{14}, p_{23}$ are outside of~$C_6$.
This is impossible, whence, there is
no circle representation of~$\AAsixMiscConVierundZwanzig$.
\smallskip

The arrangement $\AAsixMiscConAchtA$ is shown in
Figure~\ref{fig:misc_con6_sym8_2} and again in
Figure~\ref{fig:N24}(c). The proof of non-circularizability of
this arrangement is exactly as the previous proof,
just replace  $\AAsixMiscConVierundZwanzig$ by $\AAsixMiscConAchtA$
and think of an analog of Figure~\ref{fig:N24}(b).
\qed

\bigskip

Originally, we were aiming at deriving the non-circularizability of
$\AAsixMiscConVierundZwanzig$ as a corollary to the following theorem.
Turning things around we now prove it as a corollary to the
non-circularizability of $\AAsixMiscConVierundZwanzig$. We say that a polytope
$P$ has the combinatorics of the cube if $P$ and the cube have isomorphic face
lattices. The graph of the cube is bipartite, hence, we can speak of the
white and black vertices of a polytope with the combinatorics of the cube.

\begin{theorem}
Let $S$ be a sphere.
There is no polytope $P$ with the combinatorics of the cube 
such that the black vertices of $P$ are inside $S$ and the white vertices
of $P$ are outside~$S$.
\end{theorem}

\begin{proof}
Suppose that there is such a polytope $P$. Let $\EE$ be the arrangement of
planes spanned by the six faces of $P$ and let $\CC$ be the arrangement
of circles obtained from the intersection of $\EE$ and~$S$. 
This arrangement is isomorphic to $\AAsixMiscConVierundZwanzig$.
To see this, consider the eight triangles of 
$\AAsixMiscConVierundZwanzig$ corresponding to the
black and gray points of Figure~\ref{fig:N24}(b).
Triangles corresponding to black points are Krupp and
triangles corresponding to gray points are NonKrupp.
By Fact~\ref{fact:Krupp_is_inside} this translates to
corners of $P$ being outside and inside~$S$, respectively.
\end{proof}

%%%%%%%%%%%%%%%%%%%%%%%%%%%%%%%%%%%%%%%%%%%%%%%%%%%%%%%%%%%%%%%%%%%%%%%%%%%
\subsubsection{Non-circularizability of \texorpdfstring{$\AAsixMiscConAchtB$}{N6c8:2}}

The arrangement $\AAsixMiscConAchtB$ is shown in
Figure~\ref{fig:misc_con6_sym8_1} and again in
Figure~\ref{fig:Nc82}(a). 

%%%%%%%%%%%%%%%%%%%%%%%%%%%%%%%%%%%%%%%%%%%%%%%%%%%
%%
% in einem figure environment mit caption
   \calc_figscale{41}
    \begin{figure}[htb]
    \centerline{\input{\path/Nc82.pstex_t}}
    \caption{\label{fig:Nc82}}
    \end{figure}
    VC
{(a) The arrangement $\AAsixMiscConAchtB$ with four gray triangles.\\
 (b) $\AAsixMiscConAchtB$ after collapsing the gray triangles.\\
 (c) After moving the point $p$ to infinity.}	
%%
%%%%%%%%%%%%%%%%%%%%%%%%%%%%%%%%%%%%%%%%%%%%%%%%%%%

Suppose that $\AAsixMiscConAchtB$ has a
circle representation~$\CC$. Circle $C_1$ only has two triangles in
the exterior, in the figure they are gray.  Circle $C_6$ only has two
triangles in the interior.  With Lemma~\ref{lem:sl2} these four
triangles can be collapsed into points of triple intersection. 
This results in a (non-simple) arrangement as shown in Figure~\ref{fig:Nc82}(b).
Note that we do not care, whether the circles $C_1$ and $C_6$ cross or not.

Apply a M\"obius transformation that maps the point
$p = C_4 \cap C_5 \cap C_6$ to the point $\infty$ of the extended
complex plane. This maps $C_4$, $C_5$, and $C_6$ to lines, while
$C_1$, $C_2$, and $C_3$ are mapped to circles. From the order of
crossings, it follows that the situation is essentially as shown in
Figure~\ref{fig:Nc82}(c).  This Figure also shows the line $\ell_1$
through the two intersection points of $C_1$ and $C_2$, the line
$\ell_2$ through the two intersection points of $C_2$ and $C_3$, and
the line~$\ell_3$ through the two intersection points of $C_3$ and
$C_1$.  The intersection points of $\ell_3$ with the other two are
separated by the two defining points of $\ell_3$. According to
Theorem~\ref{thm:Incid-3C3L}, however, the three lines should share a
common point.  The contradiction shows that there is no circle
representation of $\AAsixMiscConAchtB$.

%%%%%%%%%%%%%%%%%%%%%%%%%%%%%%%%%%%%%%%%%%%%%%%%%%%%%%%%%%%%%%%%%%%%%%%%%%%
\subsubsection{Non-circularizability of  \texorpdfstring{$\AAsixMiscConVierA$}{N6c4:1} and 
  \texorpdfstring{$\AAsixMiscConVierB$}{N6c4:2}}
\label {ssec:extra2}

The arrangements $\AAsixMiscConVierA$ and $\AAsixMiscConVierB$ are
shown in Figure~\ref{fig:extra2} and again in
Figure~\ref{fig:extra2-arrs}. These are the only
two connected digon-free arrangements of 6 pseudocircles with a
symmetry group of order 4 which are not circularizable.
The two proofs of non-circularizability are very similar.

%%%%%%%%%%%%%%%%%%%%%%%%%%%%%%%%%%%%%%%%%%%%%%%%%%%%%%%%%%%%%%%%%%%%%%%%%%%
\begin{figure}[htb]
  \centering
  
  \hbox{}
  \hfill
    \begin{subfigure}[t]{.26\textwidth}
    \centering
    \includegraphics[page=1,width=0.95\textwidth]{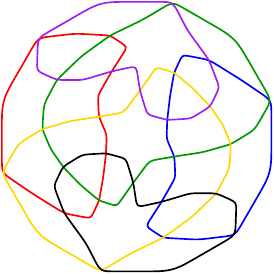}
    \caption{}
    \label{fig:misc_con6_sym4_1}
  \end{subfigure}
  \hfill
    \begin{subfigure}[t]{.26\textwidth}
    \centering
    \includegraphics[page=1,width=0.95\textwidth]{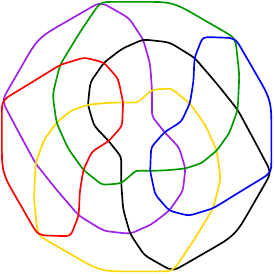}
    \caption{}
    \label{fig:misc_con6_sym4_2}
  \end{subfigure}
  \hfill
  \hbox{}
  
  \caption{Two non-circularizable arrangements of $n=6$ pseudocircles
      with a symmetry group of order 4.\\ The arrangements are denoted as
      (a)~$\AAsixMiscConVierA$ and (b)~$\AAsixMiscConVierB$.
      }
  \label{fig:extra2}
\end{figure}

%%%%%%%%%%%%%%%%%%%%%%%%%%%%%%%%%%%%%%%%%%%%%%%%%%% 
%%
% in einem figure environment mit caption
   \calc_figscale{43}
    \begin{figure}[htb]
    \centerline{\input{\path/extra2-arrs.pstex_t}}
    \caption{\label{fig:extra2-arrs}}
    \end{figure}
    VC
{(a) The arrangement $\AAsixMiscConVierA$ with four gray triangles.\\
 (b) The arrangement $\AAsixMiscConVierB$ with four gray triangles.}	
%%
%%%%%%%%%%%%%%%%%%%%%%%%%%%%%%%%%%%%%%%%%%%%%%%%%%%

Suppose that $\AAsixMiscConVierA$ has a circle representation~$\CC$.
In Figure~\ref{fig:extra2-arrs}(a) the pseudocircles $C_5$ and $C_6$
each have two gray triangles on the outside and these are the only
triangles on the outside of the two pseudocircles.  With
Lemma~\ref{lem:sl2} the respective circles in $\CC$ can be grown until
the gray triangles collapse into points of triple intersection or
until a digon flip occurs. In the case of a digon flip $C_5$ and $C_6$
become intersecting, and no further triangles incident to $C_5$ and $C_6$
are created. Therefore, it is possible to continue the growing
process until the four triangles collapse.  In the following, we do
not care whether the digon flip occured during the growth process,
i.e., whether $C_5$ and $C_6$ intersect.  The four points of triple
intersection are the points $p_1$, $p_2$, $p_3$, and $p_4$ in
Figure~\ref{fig:extra2-bew}(a).

%%%%%%%%%%%%%%%%%%%%%%%%%%%%%%%%%%%%%%%%%%%%%%%%%%% 
%%
% in einem figure environment mit caption
   \calc_figscale{43}
    \begin{figure}[htb]
    \centerline{\input{\path/extra2-bew.pstex_t}}
    \caption{\label{fig:extra2-bew}}
    \end{figure}
    VC
{(a) Illustration of the non-circularizability proof for $\AAsixMiscConVierA$.\\
 (b) Illustration of the non-circularizability proof for $\AAsixMiscConVierB$.}	
%%
%%%%%%%%%%%%%%%%%%%%%%%%%%%%%%%%%%%%%%%%%%%%%%%%%%%

There is a circle $C'_3$ which shares the points $p_1$ and $p_3$ with
$C_3$ and also contains the point $q_1$, which is defined as the
intersection point of $C_1$ and $C_5$ inside $C_3$. Similarly there is
a circle $C'_4$ which shares the points $p_2$ and $p_4$ with $C_4$ and
also contains the point $q_2$, which is defined as is the intersection
point of $C_2$ and~$C_5$ inside~$C_4$.  
By construction $C_5$ is incident to 
one intersection point of each of the pairs $C_1,C'_3$, and~$C'_3,C_2$,
and~$C_2,C'_4$ and $C'_4,C_1$.  Miquel's Theorem
(Theorem~\ref{thm:Miquels}) implies that there is a circle $C^*$
through the second intersection points of these pairs.  It can be
argued that on $C^*$ the two points $p_3$ and $p_4$ separate~$q_3$ and
$q_4$. The circle $C_6$ shares the points $p_3$ and $p_4$ with $C^*$
and contains the crossing of the pairs~$C_1,C_3$ and $C_2,C_4$ which
are `close to' $q_3$ and $q_4$ in its interior (the two points are
emphasized by the arrows in the figure). Hence $p_3$ and $p_4$ are
separated by $q_3$ and $q_4$. 
This is impossible, whence,
$\AAsixMiscConVierA$ is not circularizable.
\medskip

Suppose that  $\AAsixMiscConVierB$ has a circle representation~$\CC$.
In Figure~\ref{fig:extra2-arrs}(b) the pseudocircles $C_3$ and $C_4$ each have
two gray triangles on the outside and these are the only triangles
on the outside of the two pseudocircles. By Lemma~\ref{lem:sl2} the
respective circles in $\CC$ can be grown to make the gray triangles collapse
into points of triple intersection (we do not care whether during the
growth process $C_3$ and $C_4$ become intersecting). 
The four points of triple intersection are the points $p_1$, $p_2$,
$p_3$, and $p_4$ in Figure~\ref{fig:extra2-bew}(b).

There is a circle $C'_3$ which shares the points $p_1$ and $p_3$ with
$C_3$ and also contains the point $q_1$, this point $q_1$ is the
intersection point of $C_2$ and $C_5$ inside $C_3$. Similarly there is
a circle $C'_4$ which shares the points $p_2$ and $p_4$ with $C_4$ and
also contains the point $q_2$, this point $q_2$ is the intersection
point of~$C_1$ and~$C_5$ inside~$C_4$.  By construction $C_5$ contains
one intersection point of the pairs $C_1,C'_3$, and $C'_3,C_2$, and
$C_2,C'_4$ and $C'_4,C_1$.  Miquel's Theorem
(Theorem~\ref{thm:Miquels}) implies that there is a circle $C^*$
through the second intersection points of these pairs.  It can be
argued that on $C^*$ the points $q_3$ and $q_4$ belong to the same of
the two arcs defined by the pair $p_3,p_4$.  The circle $C_6$ shares the points
$p_3$ and $p_4$ with $C^*$ and has the crossing of the pair~$C_2,C_3$
which is outside $C^*$ in its inside and the crossing of the pair
$C_1,C_4$ which is inside $C^*$ in its outside (the two points are
emphasized by the arrows in the figure).  This is impossible, whence,
$\AAsixMiscConVierB$ is not circularizable.

%%%%%%%%%%%%%%%%%%%%%%%%%%%%%%%%%%%%%%%%%%%%%%%%%%%%%%%%%%%%%%%%%%%%%%%%%%%%% 
\section{Enumeration and Asymptotics}
\label{sec:computer-part}

Recall from Section~\ref{sec:introduction}, that the \emph{primal graph} of
a connected arrangement of $n \ge 2$ pseudocircles is the plane graph whose
vertices are the crossings of the arrangement and edges are pseudoarcs, i.e.,
pieces of pseudocircles between consecutive crossings.  The primal graph of an
arrangement is a simple graph if and only if the arrangement is digon-free.

The \emph{dual graph} of a connected arrangement of $n \ge 2$ pseudocircles is
the dual of the primal graph, i.e., vertices correspond to the faces of the
arrangement and edges correspond to pairs of faces which are adjacent along a
pseudoarc.  The dual graph of an arrangement is simple if and only if the
arrangement remains connected after the removal of any pseudocircle.  In
particular, dual graphs of intersecting arrangements are simple.

The \emph{primal-dual graph} of a connected arrangement of $n \ge 2$
pseudocircles has three types of vertices; the vertices correspond to
crossings, pseudoarcs, and faces of the arrangement.  Edges represent
incident pairs of a pseudoarc and a crossing, and of a pseudoarc and a face\footnote{
In retrospect,
we believe that adding the face-crossing edges to the primal-dual graph 
would have simplified the approach with respect to theory as well as to implementation.
}.

Figure~\ref{fig:example_p-d-pd-graph} shows the primal graph, the dual graph, and the primal-dual graph of the NonKrupp arrangement.

%%%%%%%%%%%%%%%%%%%%%%%%%%%%%%%%%%%%%%%%%%%%%%%%%%%%%%%%%%%%%%%%%%%
\begin{figure}[htb]
  \centering
  
  \begin{subfigure}[t]{.23\textwidth}
    \centering
    \includegraphics[page=1]{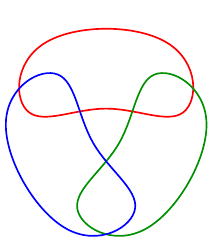}
    \caption{}
    \label{fig:example_arrangement}
  \end{subfigure}
  \hfill
  \begin{subfigure}[t]{.23\textwidth}
    \centering
    \includegraphics[page=3]{example}
    \caption{}
    \label{fig:example_p-graph}  
  \end{subfigure}
    \hfill
  \begin{subfigure}[t]{.23\textwidth}
    \centering
    \includegraphics[page=4]{example}
    \caption{}
    \label{fig:example_d-graph}  
  \end{subfigure}  
  \hfill
  \begin{subfigure}[t]{.23\textwidth}
    \centering
    \includegraphics[page=2]{example}
    \caption{}
    \label{fig:example_pd-graph}  
  \end{subfigure}
  
  \caption{(a) arrangement (b) 
    primal graph (c) dual graph (d) primal-dual graph.}
  \label{fig:example_p-d-pd-graph}
\end{figure}
%%%%%%%%%%%%%%%%%%%%%%%%%%%%%%%%%%%%%%%%%%%%%%%%%%%%%%%%%%%%%%%%%%%

We now show that crossing free embeddings of dual graphs and primal-dual
graphs on the sphere are unique.  This will allow us to disregard the
embedding and work with abstract graphs.

If $G$ is a subdivision of a 3-connected graph~$H$, then we call $G$
\emph{almost 3-connected}. If $H$ is planar and, hence, has a unique
embedding on the sphere, then the same is true for~$G$.

In the dual and the primal-dual graph of an arrangement of pseudocircles the
only possible 2-separators are the two neighboring vertices of a vertex
corresponding to a digon.  It follows that these graphs are
almost 3-connected. We conclude the following.

\begin{proposition}
\label{prop:dual_graph_unique_embedding}
  The dual graph of a simple intersecting arrangement of $n\ge 2$
  pseudocircles has a unique embedding on the sphere.
\end{proposition}

\begin{proposition}
  The primal-dual graph of a simple connected arrangement of $n\ge 2$
  pseudocircles has a unique embedding on the sphere.
\end{proposition}

Note that the statement of Proposition~\ref{prop:dual_graph_unique_embedding} 
clearly holds for $n=2$, where the dual graph is the 4-cycle.

%%%%%%%%%%%%%%%%%%%%%%%%%%%%%%%%%%%%%%%%%%%%%
\subsection{Enumeration of Arrangements}
\label{ssec:enumeration_of_arrangements}

The database of all intersecting arrangements of up to $n=7$ pseudocircles was
generated with a recursive procedure.  Arrangements of pseudocircles were
represented by their dual graphs.  The recursion was initiated with the unique
arrangement of two intersecting pseudocircles. Given the dual of an
arrangement we used a procedure which generates all possible extensions by one
additional pseudocircle.  The procedure is based on the observation that a
pseudocircle in an intersecting arrangement of $n$ pseudocircles corresponds to
a cycle of length $2n-2$ in the dual graph. 
A problem is that an arrangement of $n$
pseudocircles is generated up to $n$ times.  Since the embedding of the dual
graph is unique, we could use the canonical labeling provided by the
Graph-package of SageMath~\cite{sagemath_website} to check whether an
arrangement was found before\footnote{We recommend the Sage Reference Manual
  on Graph Theory~\cite{sagemath_graphtheory_manual} and its collection of
  excellent examples.}.

Another way for obtaining a database of all intersecting
arrangements of $n$ pseudocircles for a fixed value of $n$, is to start with
an arbitrary intersecting arrangement of $n$ pseudocircles and then perform a
recursive search in the flip-graph using the triangle flip operation (cf.\
Subsection~\ref{ssec:flipgraph_results}).

Recall that the dual graph of a connected arrangement contains multiple edges
if the removal of one of the pseudocircles disconnects the arrangement.
Hence, to avoid problems with non-unique embeddings, we modeled connected
arrangements with their primal-dual graphs. To generate the database of all
connected arrangements for $n\leq 6$, we used the fact that the flip-graph is
connected, when both triangle flips and digon flips are used
(cf.~Subsection~\ref{ssec:HS-flips}). The arrangements were created with a
recursive search on the flip-graph.  

%%%%%%%%%%%%%%%%%%%%%%%%%%%%%%%%%%%%%%%%%%%%%
\subsection{Generating Circle Representations}
\label{ssec:representations}

Having generated the database of arrangements of pseudocircles, we were then
interested in identifying the circularizable and the non-circularizable ones.

Our first approach was to generate arrangements of circles $C_1,\ldots,C_n$
with centers $(x_i,y_i)$ and radius $r_i$ by choosing triples $x_i,y_i,r_i$ at
random from $\{1,\ldots,K\}$ for a fixed constant $K \in \mathbb{N}$. In the
database the entries corresponding to the generated arrangements were marked
circularizable. Later we used known circle representations to find new ones by
perturbing values.  In particular, whenever a new circle arrangement was
found, we tried to locally modify the parameters to obtain further new
ones. With these quantitative approaches we managed to break down the list for
$n=5$ to few ``hard'' examples, which were then treated ``by hand''.

For later computations on $n=6$ (and $n=7$), we also used the information from
the flip-graph on all arrangements of pseudocircles.  In particular, to find
realizations for a ``not-yet-realized'' arrangement, we used neighboring
arrangements which had already been realized for perturbations. This approach
significantly improved the speed of realization.

Another technique to speedup our computations was to use floating point
operations and, whenever a solution suggested that an additional arrangement
is circularizable, we verified the solution using exact arithmetics.  Note
that the intersection points of circles, described by integer coordinates and
integer radii, have algebraic coordinates, and can therefore be represented by
minimal polynomials.  All computations were done using the computer algebra
system SageMath~\cite{sagemath_website}\footnote{
For more details, we refer to the
Sage Reference Manual on Algebraic Numbers and Number
Fields~\cite{sagemath_algebraicnumbers_manual}.}.

As some numbers got quite large during the computations, we took efforts to
reduce the ``size'' of the circle representations, i.e., the maximum over all
parameters $|x_i|,|y_i|,r_i$.  It turned out to be effective to scale circle
arrangements by a large constant, perturb the parameters, and divide all
values by the greatest common divisor.  This procedure allowed to reduce the
number of bits significantly when storing the circle $(x-a)^2+(y-b)^2 = r^2$
with $a,b,r \in \mathbb{Z}$.

\subsection{Counting Arrangements}

Projective arrangements of pseudolines are also known as projective abstract
order types or oriented matroids of rank~3. The precise numbers of such
arrangements are known for $n\leq 11$, 
see~\cite{Knuth1992,AichholzerAurenhammerKrasser2001,Krasser2003,AichholzerKrasser2006}. 
Hence the
numbers of great-pseudocircle arrangements given in Table~\ref{table:numbers}
are not new.  Moreover, it is well-known that there are $2^{\Theta(n^2)}$
arrangements of pseudolines and only $2^{\Theta(n \log n)}$ arrangements of
lines~\cite{Goodman1993,FelsnerGoodman2016}.  Those bounds directy translate
to arrangements of great-pseudocircles, and, in particular, 
there are at least $2^{\Theta(n^2)}$
and $2^{\Theta(n \log n)}$ arrangements of pseudocircles and circles, respectively.
In this subsection we show that the
number of such arrangements is also bounded from above by $2^{O(n^2)}$
and $2^{O(n \log n)}$, respectively.

\begin{proposition}
\label{prop:number_pseudocircle_arrangements}
There are $2^{\Theta(n^2)}$ arrangements of $n$ pseudocircles.
\end{proposition}

\begin{proof}
  The primal-dual graph of a connected arrangement of $n$~pseudocircles is a
  plane quadrangulation on $O(n^2)$ vertices.  A quadrangulation can be
  extended to a triangulation by inserting a diagonal edge in every
  quadrangular face.  It is well-known that the number of triangulations on
  $s$ vertices is $2^{\Theta(s)}$~\cite{Tutte62}.  Hence, the number of
  connected arrangements of $n$~pseudocircles is bounded by $2^{O(n^2)}$.

Since every not necessarily connected arrangements~$\AA$ on $n$ pseudocircles
can by extended by $n$ further pseudocircles to a connected arrangement~$\AA$,
and since $O((2n)^2) = O(n^2)$, the bound $2^{O(n^2)}$ also applies to the 
number of (not necessarily connected) arrangements on $n$ pseudocircles.
\end{proof}

\begin{proposition}
There are $2^{\Theta( n \log n )}$ arrangements of $n$ circles.
\end{proposition}

The proof relies on a bound for the number of cells in an arrangement
of zero sets of polynomials (the underlying theorem is associated with
the names Oleinik-Petrovsky, Milnor, Thom, and Warren). The
argument is similar to the one given by Goodman and Pollack~\cite{Goodman1986}
to bound the number of arrangements of lines, see also
Matou\v{s}ek~\cite[Chapter~6.2]{Matousek:LDG}.

\begin{proof}
  An arrangement $\CC$ of $n$ circles on the unit sphere~$\SS$ is
  induced by the intersection of $n$ planes $\EE= \{E_1,\ldots,E_n\}$
  in 3-space with~$\SS$.  Plane $E_i$ can be described by the linear
  equation $a_ix+b_iy+c_iz+d_i = 0$ for some reals $a_i,b_i,c_i,d_i$;
  we call them the parameters of $E_i$.  Below we define a polynomial
  $P_{ijk}$ of degree 6 in the parameters of the planes, such that
  $P_{ijk}=0$ iff the three circles $E_i \cap \SS$, 
  $E_j \cap \SS$, and $E_k \cap \SS$ have a common point of
  intersection. We also define a polynomial $Q_{ij}$ of degree 8 in
  the parameters of the planes, such that $Q_{ij}=0$ iff
  the circles $E_i \cap \SS$ and $E_j \cap \SS$ touch.

  Transforming an arrangement $\CC$ into an arrangement $\CC'$ in a
  continuous way corresponds to a curve~$\gamma$ in $\RR^{4n}$ from the
  parameter vector of $\EE$ to the parameter vector of $\EE'$.
  If a triangle flip or a digon flip occurs when transforming 
  $\CC$ to~$\CC'$, then $\gamma$ intersects the zero set of a polynomial
  $P_{ijk}$ or $Q_{ij}$. Hence, all the points in a fixed cell of the
  arrangement defined by the zero set of the polymials
  $P_{ijk}$ or $Q_{ij}$ with $1\leq i< j < k\leq n$ are parameter
  vectors of isomorphic arrangements of circles. 
 
  The number of cells in $\RR^{d}$ induced 
  by the zero sets of $m$ polynomials of degree at most $D$ is
  upper bounded by $(50Dm/d)^d$ (Theorem~6.2.1 in~\cite{Matousek:LDG}).
  Consequently the number of non-isomorphic arrangements of $n$ circles,
  is bounded by $(50\cdot8\cdot2{n \choose 3}/4n)^{4n}$ which is $n^{O(n)}$. 

  \medskip 
  
  For the definition of the polynomial $P_{ijk}$ we first note (see
  e.g.~\cite[Sect.~12.3]{R-G11}) that the homogeneous coordinates of the point
  $I_{ijk}$ of intersection of the three planes $E_i,E_j,E_k$ are given
  by\footnote{ If the three planes $E_i,E_j,E_k$ intersect in a common line,
    we still take the expression as a definition for $I_{ijk}$, i.e., the
    homogeneous coordiantes are all zero.}
$$
\hbox{$\bigotimes_4$}((a_i,b_i,c_i,d_i), (a_j,b_j,c_j,d_j), (a_k,b_k,c_k,d_k)),
$$
where $\bigotimes_n$ denotes the $(n-1)$-ary analogue of the cross product in~$\RR^n$
$$
\hbox{$\bigotimes_n$}(\mathbf{v}_1,\ldots,\mathbf{v}_{n-1}) := 
\begin{vmatrix}
v_1^{(1)} & \ldots & v_{n-1}^{(1)} & \mathbf{e}_1 \\
\vdots    & \ddots & \vdots        & \vdots       \\
v_1^{(n)} & \ldots & v_{n-1}^{(n)} & \mathbf{e}_n \\
\end{vmatrix}.
$$
Each component of $I_{ijk}$ is a cubic polynomial in the
parameters of the three planes.
Since a homogeneous point
$(x,y,w,\lambda)$ lies on the unit sphere~$\SS$ if and only if 
$x^2+y^2+z^2 - \lambda^2 =0$, we get a polynomial
${P}_{ijk}$ of degree~$6$ in the parameters of the planes such
that ${P}_{ijk} = 0$ iff $I_{ijk} \in \SS$.

To define the polynomial $Q_{ij}$ we need some geometric
considerations.  Note that the two circles $E_i \cap \SS$ and $E_j \cap \SS$ touch if
and only if the line $L_{ij} = E_i\cap E_j$ is tangential to
$\SS$. Let $E^*_{ij}$ be the plane normal to $L_{ij}$ which contains
the origin. The point $I^*_{ij}$ of intersection of the three planes
$E_i,E_j,E^*_{ij}$ is on $\SS$ if and only if $E_i \cap \SS$ and
$E_j \cap \SS$ touch.

A vector $N_{ij}$ which is parallel to $L_{ij}$ can be obtained as
$\bigotimes_3((a_i,b_i,c_i) , (a_j,b_j,c_j))$.  The components of
$N_{ij}$ are polynomials of degree 2 in the parameters of the planes
in $\EE$. The components of $N_{ij}$ are the first three parameters
of $E^*_{ij}$; the fourth parameter is zero. The homogeneous
components of $I^*_{ij}$ are obtained by a $\bigotimes_4$ from the
parameters of the three planes $E_i$, $E_j$, and $E^*_{ij}$. Since the
parameters of $E^*_{ij}$ are polynomials of degree 2, 
the components of $I^*_{ij}$
are polynomials of degree 4. Finally we have to test whether $I^*_{ij}\in\SS$,
this makes $Q_{ij}$ a polynomial of degree 8 in the parameters of the
planes.
\end{proof}

\subsection{Connectivity of the Flip-Graph}
\label{ssec:flipgraph_results}

Given two arrangements of $n$ circles, 
one can continuously transform one arrangement into the other.
During this transformation (the combinatorics of) the arrangement changes 
whenever a triangle flip or a digon flip occurs. 

Snoeyink and Hershberger~\cite{SH91} showed an analog for arrangements
of pseudocircles: Given two arrangements on $n$ pseudocircles, a
sequence of digon flips and triangles flips can be applied to
transform one arrangement into the other.  In other words, they have
proven connectivity of the \emph{flip-graph} of arrangements of $n$
pseudocircles, which has all non-isomorphic arrangements as vertices
and edges between arrangements that can be transformed into each other
using single flips.

For arrangements of (pseudo)lines, it is well-known that the triangle
flip-graph is connected (see e.g.~\cite{FeWe00}).  A triangle flip in an
arrangement of (pseudo)lines corresponds to an operation in the corresponding
arrangement of great-(pseudo)circles where two ``opposite'' triangles are fliped simultaneously.  

With an idea as in the proof of the Great-Circle Theorem
(Theorem~\ref{thm:PGC_theorem}) we can show that the flip-graph for
arrangements of circles is connected.

\begin{theorem}\label{thm:intersecting_circle_flipsgraph_connected}
  The triangle flip-graph on the set of all intersecting
  (digon-free) arrangements of $n$ circles is connected 
  for every $n\in\mathbb{N}$. 
\end{theorem}
\begin{proof}
  Consider an intersecting arrangement of circles $\CC$ on the unit sphere~$\SS$.
  Imagine the planes of $\EE(\CC)$ moving towards the origin. 
  To be precise, for time $t \ge 1$ let $\EE_t := \{ \nicefrac{1}{t} \cdot E : E \in \EE(\CC) \}$.
  During this process only triangle flips occur as the arrangement is
  already intersecting and  
  eventually the point of intersection of any
  three planes of $\EE_t$ is in the interior of the unit sphere~$\SS$.
  Thus, in the circle arrangement obtained by intersecting
  the moving planes $\EE_t$ with~$\SS$ 
  every triple of circles forms a Krupp,
  that is, the arrangement becomes a great-circle arrangement.  
  Since the triangle flip-graph of line arrangements is connected, 
  we can use triange flips to get to any other great-circle arrangement. 
  Note that, due to Fact~\ref{fact:digon_characterization},  
  no digon occurs in the arrangement during the whole process
  in the setting of digon-free arrangements.

  Consequently, any two arrangements of circles $\CC$ and $\CC'$ can be
  flipped to the same great-circle arrangement without digons to
  occur, and the statement follows.
\end{proof}

Based on the computational evidence for $n\leq 7$, we 
conjecture that the following is true.

\begin{conjecture}\label{conj:connected_flipgraph_intersecting_digonfree}
  The triangle flip-graph on the set of all intersecting
  (digon-free) arrangements of $n$ pseudocircles is connected 
  for every $n\in\mathbb{N}$.
\end{conjecture}

%%%%%%%%%%%%%%%%%%%%%%%%%%%%%%%%%%%%%%%%%%%%%%%%%%%%%%%%%%%%%%%%%%%

\section{Further Results and Discussion}
\label{sec:further-results}
% 
% \subsection{Circularizability}
\label{ssec:discussion_circularizability_6}

In course of this paper, we generated circle representations or proved
non-circularizability for all connected arrangements of $n \le 5$
pseudocircles (cf.~Section~\ref{sec:non-circ5}) and for all digon-free
intersecting arrangements of $n \le 6$ pseudocircles
(cf.~Section~\ref{sec:non-circ5}).  Besides that, we also investigated the
next larger classes and found
\begin{itemize}
 \item 
about 4\,400 connected digon-free arrangements of 6 circles (which is about 98\%), 
 \item 
about 130\,000 intersecting arrangements of 6 circles (which is about 90\%), and
 \item 
about 2 millions intersecting digon-free arrangements of 7 circles (which is about 66\%).
\end{itemize}
For our computations (especially the last two additional items), 
we had up to 24 CPUs running over some months
with the quantitative realization approaches described in Section~\ref{ssec:representations}.

We further investigated arrangements that were not realized by our computer
program and have high symmetry or other interesting properties.
Non-circularizability proofs for some of these candidates were presented in
Section~\ref{sec:noncirc_misc}.  Since we have no automated procedure for
proving non-circularizability, these proofs had to be done by hand.

\begin{problem}
Find an algorithm for deciding circularizability which is practical.
\end{problem}

In the case of stretchability of arrangements of pseudolines, the method based
on final polynomials, i.e., on the Gra\ss{}mann-Pl\"ucker relations, would
qualify as being practical (cf.~\cite{BOKOWSKI199021}).

%%%%%%%%%%%%%%%%%%%%%%%%%%%%%%%%%%%%%%%%%%%%%%%%%%%%%%%%%%%%%%%%%%%
\small
\advance\bibitemsep-0pt
\let\sc=\scshape
\bibliographystyle{my-siam}
\bibliography{bibliography}

\end{document}